\documentclass{CSML}

\def\dOi{13(2:5)2017}
\lmcsheading%
{\dOi}
{1--40}
{}
{}
{Nov.~18, 2015}
{May\phantom.~12, 2017}
{}

\keywords{static analysis, inter-procedural program analysis, procedure summaries, herbrand equalities}

\usepackage[utf8]{inputenc}
\usepackage[T1]{fontenc}
\usepackage[english]{babel}
\usepackage{graphicx}
\usepackage{tikz}
\usepackage[
  bookmarks,
  pdfauthor={Stefan Schulze Frielinghaus, Michael Petter, Helmut Seidl},
  pdftitle={Inter-procedural Two-Variable Herbrand Equalities},
  pdfborder={0 0 0}
  ]{hyperref}
\usepackage{array}
\usepackage{stmaryrd,amssymb,amsmath,mathtools}
\usepackage{braket}

\usetikzlibrary{arrows,automata,calc,positioning,shapes}
\tikzstyle{cfg}=[shorten >=1pt,node distance=1.2cm,on grid,>=stealth',initial text=,auto,every state/.style={draw=blue!50,fill=blue!20,thick,minimum size=1em, inner sep=2.5pt},accepting/.style={double distance=1.5pt,outer sep=0.75pt+\pgflinewidth}]

\newcommand{\x}{\ensuremath{\mathbf{x}}}
\newcommand{\y}{\ensuremath{\mathbf{y}}}
\newcommand{\z}{\ensuremath{\mathbf{z}}}
\newcommand{\X}{\ensuremath{\mathbf{X}}}
\newcommand{\C}{\ensuremath{\mathcal{C}}}
\newcommand{\T}{\ensuremath{\mathcal{T}}}
\newcommand{\True}{\ensuremath{\top}}
\newcommand{\False}{\ensuremath{\bot}}

\newcommand{\sem}[1]{\llbracket#1\rrbracket}
\newcommand{\semT}[1]{\sem{#1}^\intercal}
\newcommand{\semR}[1]{[#1]^\intercal}
\newcommand{\Let}{\coloneqq}
\newcommand{\SLPassign}{\rightarrow}
\newcommand{\impliesSharp}[1][]{\ensuremath{\DOTSB\;\Longrightarrow^\sharp\;}}
\newcommand{\length}[1]{\lVert#1\rVert}
\newcommand{\balance}[1]{\lvert#1\rvert}

\newcommand{\sizeT}[1]{\textsf{size}(#1)}
\newcommand{\hash}{\raisebox{.7pt}{\scalebox{.8}{\#}}}
\newcommand{\cupplus}{\uplus}

\theoremstyle{plain}
\newtheorem{theorem}[thm]{Theorem}
\newtheorem{corollary}[thm]{Corollary}
\newtheorem{lemma}[thm]{Lemma}

\theoremstyle{definition}
\newtheorem{example}[thm]{Example}

\begin{document}
\title{Inter-procedural Two-Variable Herbrand Equalities}
\author[S. Schulze Frielinghaus]{Stefan Schulze Frielinghaus}
\address{Technische Universit\"at M\"unchen, Boltzmannstrasse 3, 85748 Garching, Germany}
\author[M. Petter]{Michael Petter}
\address{\vspace{-18 pt}}
\author[H. Seidl]{Helmut Seidl}
\address{\vspace{-18 pt}}
\email{\{schulzef, seidl\}@in.tum.de, petter@cs.tum.edu}
\begin{abstract}
We prove that all valid Herbrand equalities can be inter-procedurally inferred
for programs where all assignments whose right-hand sides depend on at most one variable are taken into account.
The analysis is based on procedure summaries representing the weakest pre-conditions for
finitely many generic post-conditions with template variables.
In order to arrive at effective representations for all occurring weakest pre-conditions,
we show for almost all values possibly computed at run-time, that they can be uniquely
factorized into tree patterns and a ground term.
Moreover, we introduce an approximate notion of subsumption which is effectively decidable
and ensures that finite conjunctions of equalities may not grow infinitely.
Based on these technical results, we realize an effective fixpoint iteration to infer
all inter-procedurally valid Herbrand equalities for these programs.
Finally we show that an invariant candidate with a constant number of variables, can be verified in polynomial time.
 \end{abstract}

\maketitle

\noindent
How can we infer that an equality such as $\x\doteq \y$ holds at some program point,
if the operators by which the program variables $\x$ and $\y$ are computed, do not satisfy obvious
algebraic laws? This is the case, e.g., when either very high-level operations such as {\sf sqrt}, or
very low-level operations such as bit-shift are involved or, generally,
for floating-point calculations.
Still, the equality $\x\doteq \y$ can be inferred, if $\x$ and $\y$ are computed by means of
\emph{syntactically} identical terms of operator applications. The equality then is called \emph{Herbrand} equality.
The problem of inferring valid Herbrand equalities dates back to~\cite{Cocke70} where it
was introduced as the famous \emph{value numbering} problem. Since quite a while, algorithms are known which,
in absence of procedures, infer \emph{all} valid Herbrand equalities~\cite{Kildall73,Knoop90}. These algorithms can
even be tuned to run in polynomial time, if only invariants of polynomial size are of interest~\cite{Gulwani04}.
Surprisingly, little is known about Herbrand equalities if recursive procedure calls are
allowed. In~\cite{DBLP:conf/esop/Muller-OlmSS05} it has been observed that the intra-procedural techniques can be extended
to programs with local variables and \emph{functions} --- but without global variables.
The ideas there are strong enough to generally infer all Herbrand \emph{constants}
in programs with procedures and both local and global variables, i.e.,
invariants of the form $\x \doteq t$ where $t$ is ground.
Another tractable case of invariants is obtained if only assignments are taken into account whose right-hand sides
have at most \emph{one occurrence} of a variable~\cite{Petter10Thesis}.
Thus, assignment $\x \Let f(\y,a)$ is considered while
assignments such as $\x \Let f(\y,\y)$ or $\x \Let f(\y,\mathbf{z})$ are
approximated with $\x \Let\null ?$, i.e., by an assignment of an unknown value to $\x$.
The idea is to encode ground terms as numbers. Then Herbrand equalities can be represented as polynomial equalities
with a fixed number of variables and of bounded degree.
Accordingly, techniques from linear algebra are sufficient to infer all
valid Herbrand equalities for such programs.
As a special case, Petter's class of programs from~\cite{Petter10Thesis} subsumes those programs where only \emph{unary} operators
are involved.
Such programs have been considered by~\cite{Gulwani07}.
Interestingly, the latter paper arrives at decidability by a completely different line of argument,
namely, by exploiting properties of the free monoid generated from the unary
operators.
Another avenue to decidability is to restrict the control structure of programs to be analyzed.
In~\cite{Tiwari09}, the restricted class of \emph{Sloopy} Programs is introduced where the format of
loop as well as recursion is drastically restricted.
For this class an algorithm is not only provided to decide arbitrary equalities between variables but
also disequalities.

On the other hand, when only affine numerical expressions as well as affine program invariants are of concern,
the set of valid invariants at a program point form a \emph{vector space} which can be effectively represented.
This observation is exploited in~\cite{Muller-Olm04Precise} to apply methods from linear algebra to infer
all valid affine program invariants.
These methods later have been adapted to the case where values of variables are not from a field,
but where integers will overflow at some power of 2, i.e., are taken from a modular ring. Note that
in the latter structure, some number different from 0 may be a zero divisor and thus does not have a
multiplicative inverse~\cite{Muller-Olm07Analysis}.
For some applications, an analysis of \emph{general} equalities is not necessary. In applications such as
coalescing of registers~\cite{MMOSeidlAdjointESOP08} or detection of local variables in low-level code~\cite{FlexederMPS11},
it suffices to infer equalities involving two variables only.
In the affine case, algorithms for inferring all two-variable equalities can be constructed which
have better complexities than the corresponding algorithms for general equalities~\cite{FlexederMPS11}.

The question whether or not \emph{all} inter-procedurally valid Herbrand equalities
can be inferred, is still open. Here, we consider the case of Herbrand equalities containing two variables only.
These are equalities such as $\x \doteq f(g(\y),\y,a)$, i.e., right-hand sides of equalities may contain
only a single variable, but this multiple times.
Accordingly, in programs only assignments are taken into account whose
right-hand sides contain (arbitrarily many) occurrences of at most one variable.
Our main result is that under this provision,
\emph{all} inter-procedurally valid two-variable Herbrand equalities can be inferred.

Our novel analysis is based on calculating weakest
pre-conditions for all occurring post-conditions.
Since there may be infinitely many potential post-conditions for a called procedure,
we rely on \emph{generic} post-conditions to obtain finite representations of procedure summaries.
In a generic post-condition
\emph{second-order} variables are used as place-holders for yet unknown relationships between program variables.
In the generic post-condition
\[
A(\x)\doteq B(\y)
\]
the second-order variables $A$ and $B$ take as values terms with (possibly multiple occurrences of) \emph{holes} (which we call \emph{templates}).
As pre-conditions we then get conjunctions of the following form
\[
\textstyle\bigwedge_i A(s_i) \doteq B(t_i)
\]
where each term $s_i,t_i$ contains at most one variable which might occur multiple times.
To realize our algorithm for inferring all inter-procedurally valid two-variable equalities, we thus require
\begin{itemize}
\item a method to finitely represent all occurring conjunctions of equalities,
\item a method for proving that one conjunction subsumes another conjunction, i.e., a method to detect when
the greatest fixpoint computation has terminated;
\item a guarantee that a fixpoint will be reached in finitely many steps.
\end{itemize}
Note here that the equalities occurring during the weakest pre-condition computation of a generic
post-condition may contain occurrences of second-order variables.
Thus, subsumption between conjunctions of equalities is subtly related
to second-order unification~\cite{GoldfarbTCS81}. Second-order unification asks whether a conjunction
of equalities possibly containing second-order variables is satisfiable.
Since long, it is known that generally, second-order unification is undecidable.
Undecidability of second-order unification even holds if only a single unary second-order variable
is involved~\cite{Levy00}.
In contrast, the problem of \emph{context} unification, i.e., the variant of
second-order unification where
second-order variables range over terms with single occurrences of holes only,
has recently been proven to be decidable~\cite{Jez14}.
It is worth mentioning that neither of the two cases directly applies to our application,
since we consider unary second-order variables
(as context unification) but let variables range over terms with one or multiple occurrences of holes
(differently from context unification).
To the best of our knowledge, decidability of satisfiability is still open for our case.
\begin{example}
In our case, during the {\bf WP} computation a conjunction of the following form might occur:
\[
A(a) \doteq B(f(a,a)) \land A(b) \doteq B(f(b,b))
\]
where $a$ and $b$ are atoms.
The (unique) solution for the second-order variables $A$ and $B$ is then given as
\[
A = B(f(\bullet,\bullet))
\]
where $\bullet$ denotes the hole.
Since the hole occurs two times in the solution, the conjunction is not satisfiable,
if only context unification is considered.
\qed
\end{example}

In this paper, we will not solve the satisfiability problem for the given unification problem.
Instead, we introduce two novel ideas to circumvent this problem and still
infer all inter-procedurally valid two-variable Herbrand equalities.
First, we introduce a notion of \emph{approximate} subsumption. This means that our algorithm does not allow to prove
implications between all conjunctions of equalities --- but at least sufficiently many so that accumulation
of \emph{infinite} conjunctions is ruled out.
Second, we note that subsumption is not required for arbitrary valuations of program variables.
Instead it suffices to consider values which may possibly be constructed by the program at run-time.
For programs where every right-hand side of assignments contain occurrences of single variables only,
we observe that the ground terms possibly occurring at run-time, have a specific structure, which
allows for a \emph{unique factorization} of these terms into irreducible templates ---
at least, if these ground terms are sufficiently \emph{large}.
Our factorization result applied to these kind of values, enables us to
make use of the monoidal methods of~\cite{Gulwani07}.
This approach, which works for sufficiently large terms, then is complemented with
a dedicated treatment of finitely many exceptional cases.
By that, we ultimately succeed to construct an effective approximative subsumption
algorithm which allows us to restrict the number of equalities in occurring conjunctions
and to determine all valid two-variable Herbrand equalities.

In order to arrive at our key result, namely an algorithm to infer
all valid inter-procedural two-variable Herbrand equalities,
we thus build on the following two novel technical constructions:
\begin{itemize}
\item a method to uniquely factorize the kind of values possibly occurring at run-time (except finitely many) of a given program;
\item a notion of approximative subsumption which is decidable and still
guarantees that every occurring conjunction of equalities is effectively equivalent to a finite
conjunction.
\end{itemize}
Subsequently, we sketch how not only all two-variable equalities, but \emph{all} inter-procedurally valid
Herbrand equalities can be inferred, if only all right-hand sides in assignments each contain occurrences of at most
one variable.

Finally we show that the complexity of inferring all valid two-variable Herbrand equalities
in \emph{initialization-restricted} programs is polynomial and that for \emph{unrestricted} programs, at least
verifying a given equality can be performed in polynomial time. This is remarkable in so far as the terms
encountered during the {\bf WP} computation may be exponentially deep. In order to obtain a polynomial time analysis,
we therefore follow the ideas
sketched in~\cite{Gulwani07} and provide \emph{compressed} representations for the occurring terms
which support all basic term operations
in polynomial time. Subsequently, we show that our notion of approximative subsumption is decidable in polynomial time.
Furthermore, for the multi-variable case, we show that verifying an invariant candidate
is polynomial as well (given that the number of occurrences of variables in the post-condition is bounded).

Parts of this paper have been published at the ESOP conference in 2015~\cite{ESOP2015}.
For the journal version, we have provided the following additions:
\begin{itemize}
  \item an efficient implementation of the analysis by means of compressed representations of invariants;
  \item an extended program model which supports not only global but also local variables;
  \item an explicit proof of approximative $T$-subsumption and $T$-compactness.
\end{itemize}

Our paper is organized as follows. Section~\ref{s:programs} briefly introduces our programming model.
Section~\ref{s:wp} presents our basic $\mathbf{WP}$ based approach of inferring all valid program invariants.
In Section~\ref{s:basics}, we provide general background on the cancellation and factorization properties of
terms and prove a first compactness result for equalities with template variables but no occurrences of program variables.
Additionally, in Section~\ref{s:monoid} we recapitulate equalities over a free monoid.
In Section~\ref{s:restricted} we then provide an algorithm for inferring all two-variable equalities ---
at least, for programs which are \emph{initialization-restricted}
(see Section~\ref{s:restricted} for a precise definition of this restriction).
Technically, this restriction implies that all occurring terms can be uniquely factorized into irreducible terms.
In order to arrive at an algorithm for programs which are not initialization-restricted,
we complement this approach in Section~\ref{s:general}
with a dedicated treatment of values where a unique factorization is not possible.
Section~\ref{s:addon} indicates how our methods can be extended to general Herbrand equalities.
Finally, in Section~\ref{s:polytime} we examine the complexity of our analysis.
We introduce the compressed representation of terms used by the implementation
and indicate how the required operations can be efficiently realized.
There, we first consider two-variable Herbrand equalities only and afterwards also generalize the method to multi-variable
Herbrand equalities.

 \section{Programs}\label{s:programs}

For the purpose of this paper, we consider imperative programs which consist of a finite set $P$ of procedures such as:
\[
\begin{array}{@{}>{\scriptstyle}r@{\hspace{1em}}l@{\hspace{2cm}}>{\scriptstyle}r@{\hspace{1em}}l@{}}
        0:&\textsf{global}\;\x,\y;	&&	\\
        1:& \textit{main}()\;\{        &6:& p()\;\{ \\
        2:& \qquad      \x \Let a;     &7:& \qquad\textsf{if}\;(*)\;\{ \\
        3:& \qquad      \y \Let a;     &8:& \qquad\qquad\x \Let f(\x,\x); \\
        4:& \qquad      p();           &9:& \qquad\qquad p(); \\
        5:& \}                        &10:& \qquad\qquad \y \Let f(\y,\y); \\
          &                           &11:& \qquad\} \\
          &                           &12:& \}
\end{array}
\]
Instead of operating on the syntax of programs, we prefer to represent each procedure by a
(non-deterministic) control flow graph. Figure~\ref{f:prog} shows, e.g., the control flow graphs for the given example
program.
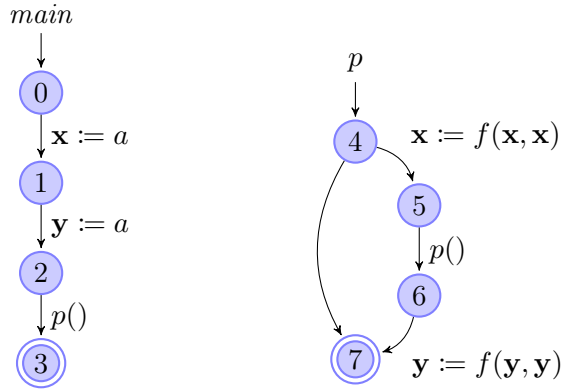
\begin{figure}[htp]
\centering
\begin{tikzpicture}[cfg]
\node[state,initial above,initial text=\textit{main}] (0) {0};
\node[state,below of=0] (1) {1};
\node[state,below of=1] (2) {2};
\node[state,below of=2,accepting] (3) {3};
\path[->] (0) edge node {$\x \Let a$} (1)
          (1) edge node {$\y \Let a$} (2)
          (2) edge node {$p()$} (3);
\end{tikzpicture}
\hspace{2cm}
\begin{tikzpicture}[cfg]
\node[state,initial above,initial text=$p$] (4) {4};
\node[state,below right of=4] (5) {5};
\node[state,below of=5] (6) {6};
\node[state,below left of=6,accepting] (7) {7};
\path[->] (4) edge[bend left] node {$\x \Let f(\x,\x)$} (5)
          (4) edge[bend right] (7)
          (5) edge node {$p()$} (6)
          (6) edge[bend left] node {$\y \Let f(\y,\y)$} (7);
\end{tikzpicture}
\caption{\label{f:prog}The corresponding CFGs for the example program.}
\end{figure}
Formally, the control flow graph for a procedure $p$ consists of:
\begin{itemize}
\item	A finite set $N_p$ of program points where $s_p,r_p\in N_p$ represent the start and return point of the
	procedure $p$;
\item	A finite set $E_p$ of edges $(u,s,v)$ where $u,v\in N_p$ are program points and $s$ denotes a basic statement.
\end{itemize}
For simplicity, we proceed in the style of Sharir/Pnueli in~\cite{SharirP81} and
consider parameterless procedures which operate on global variables only.
In the following, $\X$ denotes the finite set of program variables.
As \emph{values}, we consider uninterpreted operator expressions only.
Thus, values are constructed from atomic values by means of (uninterpreted) operator applications.
Let $\Omega$ denote a finite signature containing a non-empty set of atomic values $\Omega_0$ and sets
$\Omega_k$,$k>0$, of constructors of rank $k$. Then $\T_\Omega$ denotes the set of all possible (ground) terms over
$\Omega$, and $\T_\Omega(\X)$ the set of all possible terms over $\Omega$ and (possibly) occurrences of program
variables from $\X$.
In general, we will omit brackets around the argument of unary symbols.
Thus, we may, e.g., write $h\x$ instead of $h(\x)$.

As basic statements, we only consider assignments and procedure calls.
An assignment $\x \Let\null ?$ non-deterministically assigns \emph{any} value to the program variable $\x$, whereas
an assignment $\x \Let t$ assigns the value constructed according to the right-hand side term $t\in \T_\Omega(\X)$.
A procedure call is of the form $p()$ for a procedure name $p$.

In this paper, we only consider assignments whose right-hand sides contain occurrences of at most one variable.
The assignments occurring in the example program from Figure~\ref{f:prog} have this property.
Note that this program does not fall into Petter's class,
since the right-hand sides of assignments contain more than one occurrence of a variable.
In \emph{general} programs with arbitrary assignments, the
assignments with right-hand sides not conforming to the given restriction may, e.g.,
be abstracted by the non-deterministic assignment of \emph{any} value.

 \section{Computing Weakest Pre-conditions}\label{s:wp}

Our goal is to prove for a given assertion whether it is valid at a given program point or, better, to
infer all invariants which are valid at that point.
For that, we would like to calculate weakest pre-conditions of assertions, or, more generally, to
determine for every program point the minimal assumptions to be met
for the queried assertion to hold at the given program point.
Since the program model makes use of non-deterministic branching, we may assume
w.l.o.g.\ that every program point is \emph{reachable}.
In particular, this implies that no procedure is definitely non-terminating, i.e.,
that for every procedure $p$,
there is at least one execution path from the start point of $p$ reaching the end point of $p$.

\begin{example}\label{e:prog}
Consider the program from Figure~\ref{f:prog}.
At program exit, the invariant $\x \doteq \y$ holds.
In a proof of this fact by means of a $\mathbf{WP}$ computation,
weakest pre-conditions must be provided for procedure $p$ and all
assertions $\x \doteq t_k$, $k\geq 0$, where
$t_0 = \y$ and for $k>0$, $t_k = f(t_{k-1},t_{k-1})$.
This set of post-conditions is not only infinite,
but also makes use of an ever increasing number of variable occurrences.
Thus, an immediate encoding, e.g., into bounded degree polynomials as in
\cite{Petter10Thesis} is not obvious.
\qed
\end{example}

\noindent
In order to summarize the effect of a procedure for multiple but similar post-conditions,
we tabulate the weakest pre-conditions for \emph{generic} post-conditions only.
Generic post-conditions are assertions which contain \emph{template variables}
which later may be instantiated differently in different contexts for arriving post-conditions.
This idea has been applied, e.g., for affine equalities~\cite{Muller-Olm04Precise,MMOSeidlAdjointESOP08,FlexederMPS11},
for polynomial equalities~\cite{MMOSeidlPetterSTACS06,Petter10Thesis},
or for Herbrand equalities with unary operators~\cite{Gulwani07}.
The generic post-conditions which are of interest here, are of the forms
\[
A\x \doteq C \quad\text{or}\quad A\x \doteq B\y
\]
where $\x,\y$ are program variables, the \emph{ground} template variable $C$ is meant to receive a constant value,
and the template variables $A,B$ take \emph{templates} as values, i.e.,
terms over the ranked alphabet $\Omega$ and having at least one occurrence
of the (fresh) place holder variable $\bullet$.
Computing weakest pre-conditions operates on assertions where an assertion
is a (possibly infinite) conjunction of equalities.
The equalities occurring during weakest pre-condition calculations
are of the forms:
\[
As \doteq C \quad\text{or}\quad As \doteq Bt
\]
where $s,t$ are terms possibly containing a program variable, i.e., $s,t \in \T_\Omega(\X)$.

Consider a mapping $\sigma$ which assigns \emph{appropriate values} to the program variables from $\X$
as well as to the (non-ground or ground) template variables $A,B,C$.
This means that $\sigma$ assigns ground terms to the variables in $\X\cup\{C\}$ and templates to
$A,B$. Such a mapping is called \emph{variable assignment}.
The variable assignment $\sigma$ \emph{satisfies} the equality $s\doteq t$ ($\sigma\models(s\doteq t)$ for short)
iff $\sigma^*(s) = \sigma^*(t)$
where $\sigma^*$ is the natural tree homomorphism corresponding to $\sigma$, which is the identity on
all operators in $\Omega$.
The homomorphism $\sigma^*$ maps, e.g., the application $At$ of the template variable $A$ to the term $t$ into
$\sigma(A)[\sigma^*(t)/\bullet]$, i.e., the \emph{substitution} of the term $\sigma^*(t)$ into
the occurrences of the dedicated variable $\bullet$ in the template $\sigma(A)$.
Substitution into the dedicated variable $\bullet$ is an associative binary operation where the neutral
element is the template consisting of $\bullet$ alone. In the following, we denote this operation by juxtaposition.

Consider, e.g., an assignment $\sigma$ with $\sigma(A) = h(\bullet,\bullet)$, and $\sigma(B) = \bullet$, and
$\sigma(\x)=a$. Then
\[
\sigma^*(A\x)= h(\bullet,\bullet)\,a = h(a,a)=\bullet\,h(a,a) = \sigma^*(B h(\x,a))
\]
holds. Therefore, $\sigma$ satisfies the equality $A\x\doteq B h(\x,a)$.
In the following, we will no longer distinguish between $\sigma$ and $\sigma^*$.

The variable assignment $\sigma$ satisfies the conjunction $\phi$ of equalities
($\sigma\models\phi$ for short), iff $\sigma\models e$ for all equalities $e\in\phi$.

In our application, it will be convenient not to consider arbitrary variable assignments,
but only those which map program variables to \emph{reasonable} values as shown in the following.
For a subset $T\subseteq \T_\Omega$ of ground terms, we call a variable assignment $\sigma$ a
$T$-assignment, if $\sigma$ maps program variables $\x$ to values $\sigma(\x)\in T$ only.

The conjunction $\phi$ then is called \emph{$T$-satisfiable} if there is some $T$-assignment $\sigma$ with
$\sigma\models\phi$. Otherwise, it is \emph{$T$-unsatisfiable}.
Conjunctions $\phi,\phi'$ are \emph{$T$-equivalent} if for every $T$-assignment $\sigma$,
$\sigma\models\phi$ iff $\sigma\models\phi'$.
Obviously, an empty conjunction is satisfied by every variable assignment and therefore equal to $\True$ (true),
while all $T$-unsatisfiable conjunctions are $T$-equivalent. As usual, these are denoted by $\False$ (false).
Finally, a conjunction $\phi'$ is \emph{$T$-subsumed} by a conjunction $\phi$, if $\phi$ is $T$-equivalent to $\phi\land\phi'$.

If the set $T$ by which we have relativized the notions of satisfiability, equivalence and subsumption
equals the full set $\T_\Omega$, we may also drop the prefixing with $T$. In particular, we have for any $T$ that
satisfiability, equivalence and subsumption imply $T$-satisfiability, $T$-equivalence and $T$-subsumption, while the reverse implication may not necessarily hold.

In the following, we recall the ingredients of weakest pre-condition computation for assignments as well as
for procedure calls as provided, e.g.\ in~\cite{HoareC69} or~\cite{CousotMLPP90}.
The weakest pre-conditions of $\phi$ w.r.t.\ assignments are given by:
\[
\begin{array}{lcl}
\semT{\x \Let t}\;\phi 	&=&	\phi[t/\x]	 \\[1ex]
\semT{\x \Let\null ?}\;\phi         &=&     \forall\,\x.\;\phi	\\
\end{array}
\]
Thus, the weakest pre-condition for an assignment $\x \Let t$ is given by substitution of the term $t$ into
all occurrences of the variable $\x$ in the post-conditions, while
the weakest pre-condition for a non-deterministic assignment $\x \Let\null ?$ of any value is given
by universal quantification.
For Herbrand equalities, universal quantification can be computed as follows.
Recall that universal quantification commutes with conjunction. Therefore, it suffices to consider single equalities $e$.
If $\x$ does not occur in $e$, then $\forall\,\x.\,e$ is equivalent to $e$.
If $\x$ occurs only on one side of $e$, then $\forall\,\x.\,e = \False$.
Now assume that $\x$ occurs on both sides of $e$.
If $e$ is of the form $s\x \doteq t\x$ for templates $s,t$ (no template variables),
then either $s = t$ and hence $e$ as well as $\forall\,\x.\,e$ is equivalent to $\True$,
or $s\neq t$, in which case $\forall\,\x.\,e$ equals $\False$.
If $e$ is of the form $As\x \doteq Bt\x$ for templates $s,t$, then $\forall\,\x.\,e$ is equivalent to
$As \doteq Bt$.

Every transformation $f$ which is specified for generic post-conditions
to conjunctions of pre-conditions, can be uniquely
extended to a transformation $\bar f$ of \emph{arbitrary} post-conditions by
\[
\textstyle
\bar f(\bigwedge E) \quad=\quad \bigwedge_{e\in E} \bar f(e)
\]
where the transformation $\bar f$ for an arbitrary equality $e$ is defined as follows:
\[
\bar f(s \doteq t) =
\begin{cases}
f(A\x \doteq B\y)[s'/A,t'/B]  &\text{if $s = s'\x$, $t = t'\y$} \\
f(A\x \doteq C)[s'/A,t/C]     &\text{if $s = s'\x$, $t$ ground}	\\
f(A\x \doteq C)[s/C,t'/A]     &\text{if $t = t'\x$, $s$ ground}   \\
s \doteq t                    &\text{otherwise}
\end{cases}
\]
Subsequently, the extended function $\bar f$ is denoted by $f$ as well.
The procedure summaries are then characterized by the
constraint system {\bf S}:
\[\begin{array}{@{}lll@{\qquad}l@{}}
\semT{r_p} &\implies& \textsf{Id} & \text{for each procedure $p$}\\[1ex]
\semT{u} &\implies& \semT{s_p} \circ \semT{v} & \text{for each $(u,p(),v) \in E$}\\[1ex]
\semT{u} &\implies& \semT{s} \circ \semT{v} & \text{for each $(u,s,v) \in E$,}\\
         && & \text{$s$ assignment}
\end{array}\]
where $\circ$ means the composition of the weakest pre-condition transformers and \textsf{Id}
is the identity transformer.
Thus, accumulation of weakest pre-conditions for a generic post-condition $e$ at procedure exit $r_p$ with $e$
and then propagates its pre-conditions backward to the start point of $p$ by applying the transformations corresponding
to the traversed edges.
Here, the subsumption relation $\implies$ as defined for conjunction of equalities,
has silently been raised to the function level. Thus,
$f\implies g$ if $f(e)$ subsumes $g(e)$ for all generic post-conditions $e$.

W.r.t.\ the ordering $\sqsubseteq$ given by $\implies$, the {\bf WP} transformer
of procedure $p$ then is obtained as the value for the variable corresponding to
the start point $s_p$ in the \emph{greatest} solution to the constraint system {\bf S}.

The {\bf WP} transformers for all program points are characterized
by the greatest solution of the constraint system {\bf R}:
\[\begin{array}{@{}lll@{\qquad}l@{}}
\semR{s_\textit{main}} &\implies& \textsf{Id} \\[1ex]
\semR{s_p} &\implies& \semR{u} & \text{for each $(u,p(),\_) \in E$}\\[1ex]
\semR{v} &\implies& \semR{u} \circ \semT{s_p} & \text{for each $(u,p(),v) \in E$}\\[1ex]
\semR{v} &\implies& \semR{u} \circ \semT{s} & \text{for each $(u,s,v) \in E$,}\\
         && & \text{$s$ assignment}
\end{array}\]
The value for $\semR{v}$ for program point $v$ is meant to transform every assertion at
program point $v$, into the corresponding weakest pre-condition at the start point of the program.
Note that the constraint system for characterizing these functions makes use of the weakest pre-condition
transformers of procedures as characterized by the constraint system {\bf S}.

Assume that we are somehow given the greatest solution of the constraint system {\bf R} where $\semR{v}$ is
the corresponding transformation for program point $v$.
In order to determine all one- or two-variable equalities which are valid when reaching the program point $v$,
we conceptually proceed as follows:
\begin{description}
\item[One-variable Equality.]
	For a program variable $\x$,
	let $\psi$ denote the \emph{universal closure} of $\semR{v}(A\x \doteq C)$.
	If $\psi = \False$, then program variable $\x$ does not receive a constant
	value at program point $v$. Otherwise $\psi$ is equivalent to an equality $As \doteq C$
	where $s$ is ground, i.e., $\x \doteq s$ is an invariant at $v$.
\item[Two-variable Equality.]
	For distinct program variables $\x$ and $\y$,
	let $\psi$ denote the universal closure of $\semR{v}(A\x \doteq B\y)$.
	If $\psi = \False$, then no equality between $\x$ and $\y$ holds.
	Otherwise, $\psi$ equals a conjunction $\bigwedge_i As_i \doteq Bt_i$ of
	equalities where for each $i$ either $s_i,t_i \in \T_\Omega$ are ground
	or $s_i,t_i \in \T_\Omega(\bullet) \setminus \T_\Omega$ are templates.
	Then $r_1 \x \doteq r_2 \y$ is an invariant at $v$
	iff $r_1s_i = r_2t_i$ for all $i$, i.e.,
	any assignment $\sigma$ with $\sigma(A)= r_1,\sigma(B)=r_2$ satisfies the conjunction.
\end{description}
Here, the \emph{universal closure} of a conjunction $\phi$ is given by
$
\forall\,\x_1\ldots\forall\,\x_n.\phi
$,
if the set of program variables equals $\X=\{\x_1,\dotsc,\x_n\}$.
\begin{example}\label{e:prog1}
Consider the main procedure of the program in Section~\ref{s:programs},
as defined by the control flow graph in Figure~\ref{f:prog}.
The {\bf WP} transformer $\semR{3}$ for the endpoint $3$ of the main program
is given by:
\[
\semR{3} = \semT{\x \Let a}\circ\semT{\y \Let a}\circ\semT{4}
\]
where $4$ is the entry point of the procedure $p$.
Assume that
\[
\semT{4}(A\x\doteq B\y) = (A\x\doteq B\y) \land (Af(\x,\x)\doteq Bf(\y,\y))
\]
holds. For the program variables $\x,\y$, we therefore obtain:
\[
\begin{array}{lll}
\semR{3}(A\x\doteq B\y) &=& (A\x\doteq B\y)[a/\y][a/\x] \land (Af(\x,\x)\doteq Bf(\y,\y))[a/\y][a/\x]	\\
                        &=& (Aa\doteq Ba) \land (Af(a,a)\doteq Bf(a,a))
\end{array}
\]
This assertion does not contain occurrences of the program variables $\x,\y$.
Therefore, it is preserved by universal quantification over program variables.
Since $A = B = \bullet$ is a solution, $\x\doteq \y$ holds whenever program point $3$ is reached.
\qed
\end{example}
 \noindent
In order to turn these definitions into an effective analysis algorithm, several obstacles must be overcome.
So, it is not clear how general subsumption, as required in our characterization of the {\bf WP} transformers,
can be decided in presence of template variables.
We observe, however, that instead of general subsumption, it suffices to rely on
$T$-subsumption only --- for a well-chosen subset $T\subseteq \T_\Omega$.
Note that the smaller the set $T$ is, the coarser is the subsumption relation.
In particular for $T=\emptyset$, all conjunctions are $T$-equivalent.
Since every assertion expresses a property of reaching program states,
it suffices for our application to choose $T$ as a superset of all run-time values of program variables.

The following wish list collects properties which enable us to construct an
effective inter-procedural analysis of all two-variable Herbrand equalities:
\begin{description}
\item[$T$-Compactness.]
Every occurring conjunction $\phi$ is $T$-subsumed by a conjunction of a \emph{finite} subset
of equalities in $\phi$.
\item[Effectiveness of subsumption.]
$T$-subsumption for \emph{finite} conjunctions can be effectively decided.
\item[Solvability of ground equalities.]
The set of solutions of finite systems of equalities with template variables only, i.e., \emph{without}
occurrences of program variables can be explicitly computed.
\end{description}
By the first assumption,
a standard fixpoint iteration for the constraint systems {\bf S} and {\bf R}
will terminate after finitely many iterations (up to $T$-equivalence).
By the second assumption, termination can effectively be detected, while the third
assumption guarantees that for every program point and every program variable (pair of program variables)
the set of all valid invariants can be extracted out of the greatest solution of {\bf R}.
In total, we arrive at an effective algorithm for inferring all valid two-variable equalities.

The assumption on decidability of $T$-subsumption can be further relaxed.
Instead, we provide an \emph{approximate} notion of $T$-subsumption which is decidable.
Our approximate $T$-subsumption implies $T$-subsumption. Moreover, it
is still strong enough to guarantee that every occurring conjunction of equalities is approximately $T$-subsumed by
a finite subset of the equalities. Notions for approximate $T$-subsumption will be introduced in
Sections~\ref{s:restricted} and~\ref{s:general}.

For programs which operate on global as well as local variables, an extension of our program model and weakest
pre-condition calculus is given in Appendix~\ref{a:locals}. There we introduce a program model which is
general enough in order to model usual concepts of local variables together with call-by-value parameter parsing
and returning of results in dedicated global variables. Furthermore, we extend the weakest pre-condition calculus
in order to deal with generic post-conditions which contain local program variables.

In the upcoming section, we recall basic properties of the set of terms, possibly containing the variable
$\bullet$. These properties will allow us to deal with conjunctions of equalities where
template variables are applied to ground terms only, i.e., the case of ground equalities.

 \section{Factorization of Terms}\label{s:factorization}\label{s:basics}

Let $\T_\Omega (\bullet)$
denote the set of terms constructed from the symbols in $\Omega$, possibly together
with the dedicated variable $\bullet$.
In~\cite{Engelfriet80}, Engelfriet presents the following cancellation and
factorization properties for terms in $\T_\Omega(\bullet)$:
\begin{description}
\item[Bottom Cancellation] \hfill\\
Assume that $t_1 \neq t'_1$. Then $s_1 t_1 = s_2 t_1$ and $s_1 t'_1 = s_2 t'_1$ implies
$s_1 = s_2$.
\item[Top Cancellation] \hfill\\
Assume $\bullet$ occurs in $s$. Then $s t_1 = s t_2$ implies $t_1 = t_2$.
\item[Factorization] \hfill\\
Assume $t_i \neq t'_i$ for $i = 1, 2$. Then $s_1 t_1 = s_2 t_2$ and $s_1 t'_1 = s_2 t'_2$
implies that $s_1 r_1 = s_2 r_2$ for some $r_1,r_2$ each containing $\bullet$ where at least
one of the $r_i$ equals $\bullet$.
In that case (by top cancellation), we furthermore have that both
$r_2 t_1 = r_1 t_2$ and $r_2 t'_1 = r_1 t'_2$.
\end{description}

\noindent
Using these cancellation properties, we obtain a complete method for dealing with
equalities \emph{without} occurrences of program variables.

For one-variable equalities alone, we have the following results concerning subsumption and compactness:
\begin{theorem}\label{t:onevar}~
\begin{enumerate}
\item   A single equality $A s \doteq C$ for some ground term $s$ has exactly one solution where $A=\bullet$.
\item   Consider the conjunction
        $ As_1 \doteq C \land As_2 \doteq C  $
        for terms $s_1\neq s_2$ containing the same variable $\x$.
	If the conjunction is satisfiable, then the value of $\x$ is uniquely determined.
\end{enumerate}
\end{theorem}
\begin{proof}
We only prove the second assertion.
The conjunction $As_1 \doteq C \land As_2\doteq C$ is equivalent to the conjunction
$As_1 \doteq C \land s_1 \doteq s_2$. The most general unifier of $s_1,s_2$
maps $\x$ to a ground subterm of $s_1,s_2$ if the conjunction is satisfiable.
\end{proof}
As a consequence, we obtain:

\begin{corollary}\label{c:onevar}
Consider finite conjunctions of equalities of the form $As\doteq C$.
\begin{enumerate}
\item	Subsumption for these is decidable.
\item	Every satisfiable conjunction is equivalent to a conjunction of at most
	$n+1$ equalities where $n$ is the number of program variables.
\end{enumerate}
\end{corollary}
\noindent
Since the weakest pre-condition of a generic one-variable equality consists of equalities of the
form $As\doteq C$ only, Corollary~\ref{c:onevar} suffices to infer all inter-procedurally valid
one-variable equalities.
In the following, we therefore concentrate on the two-variable case where the weakest pre-condition
consists of conjunctions of equalities of the form $As\doteq Bt$.
First, we observe:

\begin{theorem}\label{t:solve}~
\begin{enumerate}
\item	A single equality $As \doteq Bt$ for ground terms $s,t$ has only finitely many solutions $A=r_1,B=r_2$,
	where at least one of the templates $r_1,r_2$ equals $\bullet$.
\item	Consider the conjunction
	$As_1 \doteq Bt_1 \land As_2 \doteq Bt_2$
	for ground terms $s_1\neq s_2$ and $t_1\neq t_2$. Then it has either no solution
	or there exists a unique solution $A=r_1,B=r_2$,
	where at least one of the templates $r_1,r_2$ equals $\bullet$.
	In the latter case the conjunction is equivalent to $Ar_1 \doteq Br_2$.
\item	Consider the finite conjunction $\bigwedge_{i=1}^k (As_i \doteq Bt_i)$ for ground terms $s_i,t_i$.
	Then the set of all solutions can be effectively computed,
	where at least one of the templates for $A$ or $B$ equals $\bullet$.
\end{enumerate}
\end{theorem}

\begin{proof}
For a proof of the first statement, w.l.o.g.\ assume that $s$ is at least as large as $t$.
Then for size reasons, $r_1=\bullet$. This means that $s=r_2t$ must hold.
If $t$ is not a subterm of $s$, there is no solution at all. Otherwise, i.e., if $s$ contains occurrences of
$t$, then every solution $r_2$ is obtained from $s$ by replacing a non-empty set of occurrences of $t$ with $\bullet$.

Now consider the second statement. If the pair of equalities is satisfiable then by factorization,
there are templates $r_1,r_2$ of which at least one equals $\bullet$ such that $Ar_1 \doteq Br_2$ holds.
Since at the same time $r_2s_i \doteq r_1 t_i$ holds, the equality $Ar_1 \doteq Br_2$ is equivalent to the conjunction.
Moreover, there is exactly one solution $A=r'_1,B=r'_2$ where at least one of the templates $r'_i$ equals $\bullet$,
namely, $r'_1 = r_2$, $r'_2 = r_1$.

Finally, consider the third statement. If $k = 1$, the assertion follows from statement 1.
Therefore now let $k>1$. First assume that for some $i,j$, $s_i\neq s_j$ and $t_i\neq t_j$. Then by statement 2,
the conjunction is unsatisfiable or there is exactly one pair $r_1,r_2$ of templates one of which equals $\bullet$,
such that $A=r_1, B=r_2$ is a solution of the conjunction $As_i \doteq Bt_i \land As_j \doteq B t_j$.
If in the latter case, $r_1s_l \doteq r_2 t_l$ for all $l$, we have obtained a single solution.
Otherwise, the conjunction is unsatisfiable.
Now assume that no such $i,j$ exists. Then either the conjunction is unsatisfiable or all equalities are
syntactically equal.
\end{proof}

\begin{example}
Consider the two equalities:
\[
A f(a,g b, g b) \doteq B g b \qquad
A f(a,g c, g b) \doteq B g c
\]
Then $A = \bullet$ and $B = f (a,\bullet, g b)$ is the only solution for $A,B$ where at least
one of the templates equals $\bullet$.
\qed
\end{example}
\noindent
Applying the arguments which we used to prove Theorem~\ref{t:solve},
we obtain:

\begin{corollary}\label{c:compact}
Consider a conjunction
$
\bigwedge_{i=1}^n As_i \doteq Bt_i
$
with ground terms $s_i,t_i$.
\begin{enumerate}
\item
If it is satisfiable, it is equivalent to the conjunction
of at most two conjuncts.
\item
If it is unsatisfiable, there are at most three conjuncts whose conjunction
is unsatisfiable.
\end{enumerate}
\end{corollary}

\noindent
By Theorem~\ref{t:solve}, the assumption {\bf solvability of ground equalities} from Section~\ref{s:wp} is met.
Thus, it remains to solve the constraint systems {\bf S} and {\bf R}, i.e.,
to construct an approximate $T$-subsumption relation which is both
effective and guarantees that every conjunction is approximately $T$-subsumed by the conjunction
of a finite subset of equalities.
In order to construct such a relation, we require stronger insights into the structure of templates
and their compositions.
Let $\C_\Omega$ denote the subset of all terms in $\T_\Omega(\bullet)$
which contain at least one occurrence of $\bullet$, i.e.,
$\C_\Omega = \T_\Omega(\bullet)\setminus \T_\Omega$.
The terms in $\C_\Omega$ have also been called \emph{templates}.
The set $\C_\Omega$, equipped with substitution, is a \emph{free monoid}
with neutral element $\bullet$.
This monoid consists of finite products of the 
irreducible elements in $\C_\Omega$. As usual, we call an element $t$
\emph{irreducible} if $t$ cannot be non-trivially decomposed into a product,
i.e., $t = uv$ implies that $t = u$ with $v=\bullet$ or
$t = v$ with $u=\bullet$. Note that there are \emph{infinitely} many irreducible elements in $\C_\Omega$
--- whenever $\Omega$ contains constructors of rank exceeding 1.

While templates can be uniquely factored, this is
no longer the case for ground terms, i.e., terms without variable occurrences.

\begin{example}
Consider the ground term $t = h(f(h(1),h(1)))$, together with the templates
$s_1 = h( f(\bullet,h(1)))$, $s_2 = h(f(h(1),\bullet))$ and $s_3 = h(f(\bullet,\bullet))$.
All these three templates are distinct. Still,
\begin{flalign*}
&& t = s_1 \;h(\bullet)\;1 = s_2 \;h(\bullet)\;1 = s_3\; h(\bullet)\;1 &&\null\qEd
\end{flalign*}
\end{example}

\noindent
Thus, unique factorization of arbitrary ground terms cannot be hoped for.
Still, we observe that unique factorization can be obtained --- at least up to any fixed finite set
of ground terms.
Let $G$ denote a finite set of ground terms which is closed by subterms.

Let $M_G$ denote the sub-monoid of all templates $m\in \C_\Omega$ whose ground subterms
all are contained in $G$.
Then we have:

\begin{theorem}\label{t:unique}
Assume that $S\subseteq\T_\Omega$ which is closed by subterms.
If $G\subseteq S$, then
every ground term $t\in \T_\Omega\setminus S$, can be uniquely factored into
$
	t = m x
$
such that
\begin{enumerate}
\item\label{uniqueProp1}$m\in M_G$ and $x\not\in S$;
\item\label{uniqueProp2}$x$ is minimal with property~\eqref{uniqueProp1}, i.e., there exists no $x' \in \T_\Omega \setminus S$ such that $x=sx'$ for some $s \in M_G \setminus \{\bullet\}$.
\end{enumerate}
\end{theorem}
\begin{proof}
Since $M_G \subseteq \C_\Omega$, every term in $M_G$ is uniquely factorizable.

Let $t = m_1x_1 = m_2x_2$ with $m_i \in M_G$ and $x_i \in \T_\Omega \setminus S$ are minimal according to property~\eqref{uniqueProp2} for $i=1,2$.
Then either $m_1 = m_2m'$ or $m_2 = m_1m'$ for some $m' \in M_G$ holds.
Otherwise, we have a contradiction to the assumption that $m_1x_1 = m_2x_2$ holds.
Consider the case where $m_1 \neq m_2$, i.e., $m' \neq \bullet$.
If $m_1 = m_2m'$, then we conclude that $m'x_1 = x_2$ holds.
This means, that $x_2$ is not minimal according to property~\eqref{uniqueProp2} which is a contradiction to our assumption.
A similar argument holds for $m_2 = m_1m'$.
Now consider the case where $m_1 = m_2$, then also $x_1 = x_2$ from which the assertion of the theorem follows.
\end{proof}

\begin{example}
Consider the term
\[
t = f(h(f(2,h(1))), h(f(2,h(1))))
\]
and assume that the set
$G$ of forbidden ground subterms is given by $G = \{h(1),1\}$ and $S=G$.
Then $t$ can be decomposed into:
\[
f(\bullet,\bullet)\;h(\bullet)\;f(\bullet,h(1))\;2
\]
If on the other hand, $S=G = \{2\}$, we obtain the decomposition:
\[
f(\bullet,\bullet)\;h(\bullet)\;f(2,\bullet)\;h(\bullet)\;1
\]
If finally, $S$ and $G$ are empty, the term $x$ of Theorem~\ref{t:unique} is the minimal subterm such that
the occurrences of $x$ contains all ground leaves of $t$. This means that $x = f(2,h(1))$, and
we obtain the decomposition:
\begin{flalign*}
&& f(\bullet,\bullet)\;h(\bullet)\;f(2,h(1)) &&\null\qEd
\end{flalign*}
\end{example}

The unique decomposition of the ground term $t$ claimed by Theorem~\ref{t:unique}, is constructed as follows.
Let $X$ denote the set of minimal subterms $x'$ of $t$ such that $x'\not\in G$.
Then we construct the least subterm $x\not\in S$ of $t$ such that all occurrences of subterms $x'\in X$ in $t$
are contained in some occurrence of $x$. This subterm is uniquely determined.
Then define $m$ as the term obtained from $t$ by replacing all occurrences of $x$ with $\bullet$.
This term $m$ is also uniquely determined with $t = m x$. Moreover by construction,
all ground subterms of $m$ are contained in $G$.

\begin{example}\label{e:monoid}
Consider the program from Example~\ref{e:prog}. In this program, no non-ground right-hand side contains
ground subterms. Accordingly, the set $G$ is empty. Since the only ground right-hand side equals the atom $a$,
the decomposition Theorem~\ref{t:unique} allows to uniquely decompose
all run-time values of this program into right-hand sides of assignments.
\qed
\end{example}

Theorem~\ref{t:unique} allows to extend the monoidal techniques of Gulwani et al.~\cite{Gulwani07}
for unary operators to programs where all run-time
values can be uniquely factorized into right-hand sides. This extension is
given in Section~\ref{s:restricted}.
The general case where unique factorization of all run-time values can no longer be
guaranteed, subsequently is presented in Section~\ref{s:general}.
For completeness reasons, we also present simplified versions of the algorithms for
monoidal equalities from~\cite{Gulwani07} in the next section.

 \section{Equalities over a Free Monoid}\label{s:monoid}

Consider a free monoid $M_\Sigma$ with set of generators $\Sigma$.
As usual, the neutral element of $M_\Sigma$ is denoted by $\epsilon$.
Let $F_\Sigma$ be the corresponding free group. $F_\Sigma$ can be
considered as the free monoid generated from $\Sigma \cup \Sigma^-$
(where $\Sigma^- = \{a^-\mid a\in\Sigma\}$ is the set of formal inverses of elements
in $\Sigma$ with $\Sigma\cap\Sigma^-=\emptyset$) modulo exhaustive application
of the cancellation rules $a\cdot a^- = a^-\cdot a = \epsilon$ for all
$a\in\Sigma$.
In particular, the neutral element of $F_\Sigma$ is given by $\epsilon$, and
the inverse $g^{-1}$ of an element $g=a_1\cdots a_k$, $a_i\in\Sigma \cup\Sigma^-$, is given by
$g^{-1} = a_k^{-1}\cdots a_1^{-1}$ where $x^{-1} = x^-$ and ${(x^-)}^{-1} = x$ for $x\in\Sigma$.

For every $w\in M_{\Sigma \cup\Sigma^-}$, the \emph{balance} $\balance{w}$ is the difference
between the number of occurrences of positive and negative letters in $w$, respectively.
Formally, the balance is inductively defined by
\[
  \begin{array}{lcl@{\qquad}l}
    \balance{\epsilon} &=& 0 \\
    \balance{aw} &=& \balance{w} + 1& \text{if $a \in \Sigma$} \\
    \balance{aw} &=& \balance{w} - 1& \text{if $a \in \Sigma^-$}
  \end{array}
\]
Thus, $\balance{aba^-b^-c} = 1$ and $\balance{a^-b} = 0$. Note that the balance stays invariant under
application of the cancellation rules.
Also, $\balance{uv} = \balance{u}+\balance{v}$ and $\balance{u^{-1}} = -\balance{u}$.
Accordingly, the balance $\balance{\cdot} \colon F_\Sigma\to \mathbb{Z}$ is a group homomorphism.
Furthermore, we call $w$ \emph{non-negative} if $\balance{w'} \geq 0$ for all prefixes $w'$ of $w$.
This property is also preserved by cancellation and concatenation but not by inverses.
Instead, we have:
\begin{lemma}\label{l:nonnegative}
  If both $u,v \in M_{\Sigma \cup\Sigma^-}$ are non-negative, and $\balance{u} \geq \balance{v}$ then
  also $uv^{-1}$ is non-negative.
\end{lemma}
\begin{proof}
Consider a prefix $x$ of $uv^{-1}$. If $x$ is a prefix of $u$,
$\balance{x} \geq 0$ since $u$ is non-negative. Otherwise, $x = u v'^{-1}$ for some suffix $v'$ of $v$.
Then $\balance{v'} \leq \balance{v}$, since $v$ is non-negative.
Therefore, $\balance{u v'^{-1}} = \balance{u} - \balance{v'} \geq \balance{u} - \balance{v} \geq 0$.
\end{proof}

\noindent
We consider equalities of the form:
\begin{equation}
AuA^{-1} = Bu'B^{-1}	\label{eq:base}
\end{equation}
where $A,B$ are variables which take values in $M_\Sigma$,
and $u,u' \in M_{\Sigma \cup \Sigma^-}$ are maximally canceled.
If the equality is satisfiable, then necessarily $\balance{u} = \balance{u'}$ holds.
Assume from now on that $u,u'$ are maximally canceled, and $\balance{u} = \balance{u'}$.
Furthermore, we assume that 
$u,u'$ are both non-negative.
We then have:
\begin{lemma}\label{l:base}
If $\balance{u} = \balance{u'} = 0$, then the equality~\eqref{eq:base} is either trivial, is equivalent to an
equality $As = B$ or an equality $A=Bs$ for some $s\in M_\Sigma$ or is contradictory.
\end{lemma}
\begin{proof}
Assume $u=\epsilon$. Then $B= Bu'$. Thus either $u'=\epsilon$ and the equality is trivial,
or $u'\neq\epsilon$ and the equality is contradictory.

Therefore, assume that $u\neq\epsilon\neq u'$.
Then $u$ and $u'$ must be of the form $u = xyz^{-1}$, $u' = x'y'z'^{-1}$ for maximal $x,x',z,z'\in M_\Sigma$,
i.e., $y,y'$ each are either equal to $\epsilon$ or of the form $a^-w b$ for some $a,b\in\Sigma$.
Then all $x,x',z,z'$ are different from $\epsilon$.
Then equality~\eqref{eq:base} is equivalent to:
\[
Ax = Bx' \land y=y' \land Az = Bz'
\]
By bottom cancellation, these three equalities either are equivalent to one fixed relation between $As=B$ or $A=Bs$ for some
$s\in M_\Sigma$,
or to a contradiction.
\end{proof}
\begin{example}
Consider the equality
\[ Affg^{-1}f^{-1}A^{-1} \doteq Bfg^{-1}B^{-1} \]
which is, according to Lemma~\ref{l:base}, equivalent to
\[ Aff \doteq Bf \land \epsilon \doteq \epsilon \land Afg \doteq Bg \]
By bottom cancellation, we conclude that the conjunction is equivalent to a solved equality $Af \doteq B$.
\qed
\end{example}
\noindent
Now assume that there is another equality:
\begin{equation}
AvA^{-1} = Bv'B^{-1}	\label{eq:base-two}
\end{equation}
with non-negative $v,v'$ where $\balance{v} = \balance{v'}$.
\begin{theorem}\label{t:base}
The two equalities~\eqref{eq:base} and~\eqref{eq:base-two} are effectively equivalent
either to one solved equality, or to a single equality of the form~\eqref{eq:base} or
are contradictory.
\end{theorem}
\begin{proof}
We perform an induction on the sum of balances $\balance{u} + \balance{v}$.
W.l.o.g.\ assume that $\balance{u} \geq \balance{v}$.
If $\balance{v} = 0$, then the assertion follows from Lemma~\ref{l:base}.
Therefore, assume that $\balance{v}>0$, and $r\geq 1$ is the maximal number such that
$\balance{v^r} = r\cdot \balance{v}\leq \balance{u}$.
Then we construct the elements $uv^{-r}$ and $u'v'^{-r}$, which are both non-negative
by Lemma~\ref{l:nonnegative}.
Let $w, w'$ be obtained from $uv^{-r}$ and $u'v'^{-r}$ by exhaustively applying
the cancellation rules. By construction, these are non-negative as well.
Then we consider the equality:
\begin{equation}
AwA^{-1} = Bw'B^{-1} \label{eq:base-three}
\end{equation}
which is implied by the two equalities~\eqref{eq:base} and~\eqref{eq:base-two}.

If $w=\epsilon$, then either $w'=\epsilon$ holds and the equality~\eqref{eq:base-three} is trivial,
or $w'\neq\epsilon$ and equality~\eqref{eq:base-three} is contradictory.
In the first case, the equality~\eqref{eq:base-two} is implied by equality~\eqref{eq:base},
while in the second case the two given equalities~\eqref{eq:base} and~\eqref{eq:base-two} are contradictory.
The same argument applies when $w'=\epsilon$ with the roles of $A,B$ exchanged.
Therefore now assume that $w\neq\epsilon\neq w'$.
Otherwise, the pair of equalities~\eqref{eq:base} and~\eqref{eq:base-two}
is equivalent to the pair of equalities~\eqref{eq:base-two} and~\eqref{eq:base-three},
where the sum of balances $\balance{w}+\balance{v}\leq \balance{w}+r\cdot\balance{v} = \balance{u} < \balance{u}+\balance{v}$ has decreased.
For these, the claim follows by inductive hypothesis.
\end{proof}
In~\cite{Gulwani07} a similar argument is presented. The argument there together with the resulting
algorithm has been significantly simplified by introducing the extra notion of \emph{non-negativity}.

\section{Initialization-restricted Programs}\label{s:restricted}

In the subsequent let $R$ be the set of ground right-hand sides of assignments,
and $G$ be the set of ground subterms of non-ground right-hand sides
of assignments of our program.
Then generally, each value $x$ possibly constructed at run-time by the
program is of the form $x = x'r$ where $x' \in M_G$ and $r \in R$.
\begin{lemma}
  Each program variable in $\X$ ranges over the set $M_G R$.
  \qed
\end{lemma}
This means that for pre-conditions $\phi$ possibly occurring in a {\bf WP} calculation
for a program invariant, we are only interested in variable assignments $\sigma$ which map
each program variable $\x$ to a possible run-time value for $\x$, i.e., to a value
from the set $M_G R$.
In the subsequent let
\[
  T \Let M_G R \quad\text{and}\quad T' \Let M_G \X
\]
then during the {\bf WP} computation template variables are applied to ground terms in $T$ and non-ground terms in $T'$ only.
Henceforth, we therefore no longer consider general satisfiability, equivalence and subsumption, but only
$T$-satisfiability, $T$-equivalence and $T$-subsumption.
This restriction is crucial for the generalization of the monoidal techniques from~\cite{Gulwani07}.
In the following, we first consider the sub-class of programs $p$
where set $R$ of ground right-hand sides of $p$ satisfies the two properties:
\begin{enumerate}
\item $R \cap G = \emptyset$.
\item The elements in $R$ are mutually incomparable ground terms, i.e.,
      for $r_1,r_2\in R$, $r_1$ is a subterm of $r_2$ iff $r_1= r_2$.
\end{enumerate}
The program $p$ then is called \emph{initialization-restricted} (IR for short).

\begin{example}
Assume that the non-ground right-hand sides of assignments of a program
are $f(\x,h(1))$ and $f(2,h(\y))$. Then the set $G$ is given by $G = \{1,h(1),2\}$.
A suitable set $R$ of ground right-hand sides might be, e.g., $R = \{0,a\}$.
\qed
\end{example}

Our condition here is not as restrictive as it might seem.
Programs where each variable is initialized by a non-deterministic assignment,
are all IR\@.
The same holds true for programs where all non-ground right-hand sides of
assignments do not contain ground terms, and variables are initialized with atoms only.
The latter property is met by our Example~\ref{e:prog}.
By suitably massaging variable initializations, it also comprises all programs using monadic
operators only (as in~\cite{Gulwani07}).

We distinguish between two-variable equalities of the following formats:
\[
\begin{array}{@{}l@{\qquad}lcl@{\qquad}l@{}}
{[}F_{\x,\y}{]}         & As\x &\doteq& Bt\y        & \text{where $s,t \in M_G$}\\
{[}F_{\cdot,\x}{]}      & As &\doteq& Bt\x          & \text{where $s \in T$ and $t \in M_G$}\\
{[}F_{\x,\cdot}{]}      & At\x &\doteq& Bs          & \text{where $s \in T$ and $t \in M_G$}\\
\end{array}
\]
For each format separately, we observe:
\begin{theorem}\label{t:ressubcompact}~
\begin{description}
\item[$T$-subsumption.]
For finite sets $E,E'$ of two-variable equalities of the same format
it is decidable whether $\bigwedge E$ $T$-subsumes $\bigwedge E'$ or not.
\item[$T$-compactness.]
Every $T$-satisfiable conjunction of a set $E$ of two-variable equalities of the same format is $T$-subsumed
by a conjunction of a subset of at most three equalities in $E$.
\end{description}
\end{theorem}
\begin{proof}
In order to prove the theorem we show that every $T$-satisfiable conjunction of equalities of the same format
is effectively $T$-subsumed by a conjunction of at most three equalities.
Furthermore, the proof indicates that, given three equalities, it can be effectively decided whether
or not a fourth equality is $T$-subsumed or not.
We consider one case of the assertion of the theorem after the other.

\textbf{Same variable on both sides.}
Consider the two distinct equalities
\[
As_1\x \doteq Bt_1\x\qquad As_2\x \doteq Bt_2\x
\]
where $s_i,t_i\in M_G$, and assume that the conjunction of them is $T$-satisfiable.
We claim that then
$s_1\x\neq s_2\x$ and $t_1\x\neq t_2\x$.
For that, we convince ourselves first that
$s_1\neq s_2$ and $t_1\neq t_2$ must hold.
Then for a contradiction, assume that $s_1\x \doteq s_2\x$.
Since $s_1\neq s_2$, their unifier must map $\x$ to a ground term of $s_1$ and $s_2$. These ground terms
are all contained in $G$, whereas we only consider values for $\x$ in $M_G R$, which is
disjoint from $G$.
A similar argument also shows that $t_1\x\neq t_2\x$ holds.
Thus by factorization, $Ar_1\doteq Br_2$ must hold for some $r_1,r_2 \in M_G$ of which at least one equals $\bullet$.
Due to unique factorization, we then may cancel $\x$ on both sides, resulting
in the equalities $As_1 \doteq Bt_1$ and $As_2 \doteq Bt_2$. These can be simplified
to one equality $Ar_1 \doteq Br_2$ for some $r_1,r_2\in M_G$ where $r_i=\bullet$ for at least one $i$.
Hence, the second equality is $T$-subsumed by the first one.

\textbf{One-sided single variable.}
Consider the three distinct equalities
\[
As_1 \doteq Bt_1\x\qquad As_2 \doteq Bt_2\x\qquad As_3 \doteq Bt_3\x
\]
where $s_i\in M_GR$ and $t_i\in M_G$, and assume that the conjunction of them is $T$-satisfiable.
Again, we argue that all $s_i$ must be distinct as well as all $t_i\x$.
Then again by factorization, $Ar_1 \doteq Br_2$ for some templates $r_1,r_2$ of which at least one equals $\bullet$.
By unique factorization, $s_1 = s'_1r$ for some $s'_1\in M_G$ and $r\in R$.
Therefore, again by unique factorization, the value for $\x$ also must terminate
in the term $r$, i.e., is of the form $\x = x'r$ for some $x' \in M_G$.
Accordingly, also $s_2,s_3$ can be factored as $s_i = s'_ir$ for suitable $s'_i\in M_G$.
Canceling out the ground terms $r$, we obtain the monoid equalities:
\[
As'_1 \doteq Bt_1x'\qquad As'_2 \doteq Bt_2x'\qquad As'_3 \doteq Bt_3x'
\]
Assume w.l.o.g., that the balance of $s_1$ is less or equal to the balances of $s_2$ and $s_3$.
Then the conjunction of the three equalities is $T$-equivalent to:
\[
As'_1 \doteq Bt_1x'\qquad
As'_2{s'_1}^{-1}A^{-1} \doteq Bt_2t_1^{-1} B^{-1} \qquad
As'_3{s'_1}^{-1}A^{-1} \doteq Bt_3t_1^{-1} B^{-1}
\]
where $s'_2{s'_1}^{-1},t_2t_1^{-1}, s'_3{s'_1}^{-1}, t_3t_1^{-1} $ all are non-negative.
According to Theorem~\ref{t:base}, the two last equalities are either
$T$-equivalent to each other, which means that the initial conjunction is $T$-equivalent to the conjunction of the two equalities
\[
As_1 \doteq Bt_1\x\qquad As_2 \doteq Bt_2\x
\]
and the assertion follows.
Otherwise, they are $T$-equivalent to an equality $Ar_1 \doteq Br_2$ for templates $r_1,r_2$ of which at least one equals $\bullet$.
A fourth equality is then either $T$-subsumed or falsifies the conjunction of equalities.
A similar argument applies to equalities of the form $At_i\x \doteq Bs_i$.

\textbf{Different variables on both sides.}
Consider the three distinct equalities
\[
As_1\x \doteq Bt_1\y\qquad As_2\x \doteq Bt_2\y\qquad As_3\x \doteq Bt_3\y
\]
for distinct program variables $\x,\y$
where $s_i,t_i\in M_G$, and assume that the conjunction of them is $T$-satisfiable.
As before, we argue that
$s_i\x\neq s_j\x$,
$t_i\y\neq t_j\y$ for all $i\neq j$ must hold.
Then by factorization, $A$ is a prefix of $B$ or vice versa.
But then, due to unique factorization, also $As_1$ is a prefix of $Bt_1$ or vice versa.
This means that there are ${\bf u},{\bf v}\in M_G$ of which one equals $\bullet$ such that
$As_1{\bf u} \doteq Bt_1{\bf v}$, which (by top cancellation) implies that ${\bf v}\x = {\bf u}\y$ holds.
From that, we conclude that $As_i{\bf u} \doteq Bt_i{\bf v}$ for all $i$.
Assume again w.l.o.g.\ that the balance of $s_1$ is less or equal to the balances of $s_2$ and $s_3$.
We then proceed as in the last case to obtain the $T$-equivalent three equalities:
\[
As_1{\bf u} \doteq Bt_1{\bf v} \qquad
As_2s_1^{-1}A^{-1} \doteq Bt_2t_1^{-1} B^{-1} \qquad
As_3s_1^{-1}A^{-1} \doteq Bt_3t_1^{-1} B^{-1}
\]
where $s_2s_1^{-1},t_2t_1^{-1}, s_3s_1^{-1}, t_3t_1^{-1} $ all are non-negative.
According to Theorem~\ref{t:base}, the latter two equalities again are $T$-equivalent to
an equality $Ar_1 \doteq Br_2$ for templates $r_1,r_2$ of which at least one equals $\bullet$,
or are $T$-equivalent to each other, and the assertion of the theorem follows.
This completes the proof.
\end{proof}
 
It relies on the unique factorization property together with the monoidal techniques from Section~\ref{s:monoid}.
Since $T$-subsumption is decidable, at least for equalities of the same format,
we define an approximate $T$-subsumption relation $\bigwedge E \impliesSharp\bigwedge E'$
for conjunctions of equalities as follows.
Let $E_F$ and $E_F'$ denote the subsets of equalities of the same format $F$ in $E$ and $E'$, respectively.
Then $\bigwedge E \impliesSharp\bigwedge E'$ holds iff
$\bigwedge E_F$ $T$-subsumes $\bigwedge E_F'$ for all formats $F$.
Hence, by Theorem~\ref{t:ressubcompact}, we obtain:

\begin{corollary}\label{c:ressubcompact}
Assume that $n$ is the number of program variables.
\begin{description}
\item[Approximate $T$-subsumption.]
For finite sets $E,E'$ of two-variable equalities,
it is decidable whether $\bigwedge E$ approximately $T$-subsumes $\bigwedge E'$ or not.
\item[Approximate $T$-compactness.]
Every $T$-satisfiable conjunction of a set $E$ of two-variable equalities
is approximately $T$-subsumed by a conjunction of a subset of at most $\mathcal{O}(n^2)$ equalities in $E$.
\end{description}
\end{corollary}
Overall, we therefore conclude for IR programs:
\begin{theorem}\label{t:restricted}
Assume that $p$ is an IR program. Then for every program point $u$,
the set of all two-variable equalities can be determined that are valid when reaching
program point $u$.
\end{theorem}
\begin{proof}
By Corollary~\ref{c:ressubcompact}, the greatest solutions of the constraint systems {\bf S} and {\bf R} can be
effectively computed. Let $\semR{u}$, $u$ program point, denote the greatest solution of the system {\bf R}.
Then the set of valid equalities $s\x\doteq t\y$ between program variables $\x$, $\y$ is given by the
set of solutions to a system of ground equalities which are obtained by universal quantification
over all program variables of the conjunction of equalities $\semR{u}(A\x\doteq B\y)$.
By Theorem~\ref{t:solve}, a representation of the set of solutions for the template variables $A,B$ in
this conjunction can be explicitly computed.
Likewise, the set of valid equalities $x\doteq t$ for program variable $\x$ and ground term $t$
can be extracted from the universal quantification
over all program variables of the conjunction of equalities $\semR{u}(A\x\doteq C)$.
The resulting conjunction may either equal $\False$ (no constant value for $\x$) or contain only the variable $C$.
Consequently, the possible constant value for $\x$ and program point $u$
can also be effectively computed.
This completes the proof.
\end{proof}

\begin{example}\label{e:prog2}
According to our constructions in Section~\ref{s:wp} and Theorem~\ref{t:solve},
the set of all inter-procedurally valid assertions can be obtained from
the greatest solutions to the constraint systems {\bf S} and {\bf R}.
Consider, e.g., the constraint system {\bf R} for the recursive procedure $p$ from Section~\ref{s:programs},
as defined by the control flow graph of Figure~\ref{f:prog}.
If Round-Robin iteration is applied to calculate the transformers $\semT{u}$ for the
program points $u = 4,5,6,7$, we obtain for the generic post-condition $A\x\doteq B\y$
the result depicted by Table~\ref{t:RRiter}
\begin{table}
\caption{\label{t:RRiter}Round-Robin iteration for the procedure $p$ from Figure~\ref{f:prog}}
\centering$\begin{array}{|l||lll|lll|lll|}
\hline
	& &1& \  & & 2 & \ & & 3 & \ 	\\
\hline\hline
7	& A\x&\doteq &B\y	&&&		\	&&&	\ \\
\hline
6	& A\x&\doteq &B f(\bullet,\bullet)\y&&&	\	&&&	\ \\
\hline
5	& &\True& \	&		A\x&\doteq &B f(\bullet,\bullet)\y&
					Af(\bullet,\bullet)\x&\doteq &B f(\bullet,\bullet) f(\bullet,\bullet)\y\\
\hline
4	& A\x&\doteq &B\y &		Af(\bullet,\bullet)\x&\doteq &B f(\bullet,\bullet)\y&
					Af(\bullet,\bullet)f(\bullet,\bullet)\x&\doteq &
					B f(\bullet,\bullet) f(\bullet,\bullet)\y\\
\hline
\end{array}$
\end{table}
where in the $i$th column, we have only displayed pre-conditions
which have additionally been attained in the $i$th iteration for the program points
$7,6,5$ and $4$, respectively.
For convenience, we have displayed the terms in equalities according to their unique factorizations.
For program point $4$, the two equalities after the second iteration,
imply:
\[
Af(\bullet,\bullet)A^{-1} \doteq Bf(\bullet,\bullet)B^{-1}
\]
The second equality for program point $4$ together with this identity imply that
\[
	Af(\bullet,\bullet)A^{-1}Af(\bullet,\bullet)\x \doteq Bf(\bullet,\bullet)B^{-1}Bf(\bullet,\bullet)\y
\]
from which the third equality for program point $4$ as provided by the
third iteration follows.
Thus, Round-Robin fixpoint iteration reaches the greatest fixpoint after the third iteration.
\qed
\end{example}
  
\section{Unrestricted Programs}\label{s:general}

Our analysis of IR programs relied on the fact that all run-time values
of program variables can be uniquely factorized.
This was possible since in IR programs the ``bottom end'' of values
can be uniquely identified by means of the ground right-hand sides from $R$.
In general, though, ground right-hand sides could very well also occur as subterms
of other right-hand sides in the program.
In this case, we can no longer assume that $R$
serves as such a handy set of end marker terms.
At first sight, therefore, the monoidal method seems no longer applicable.
A second look, however, reveals that the monoidal method essentially fails
only, where program variables take \emph{small} values.
Again, let $R$ and $G$ denote the set of all ground right-hand sides and the set
of all ground subterms of non-ground right-hand sides of assignments in the program, respectively.
We call a term in $M_G R$ \emph{small} if it is a ground subterm of a right-hand side
of an assignment. Let us denote the (finite) set of all small terms by $S$.
Thus in particular, $R\subseteq S$.
The terms in $M_G R$ which are not small, are called \emph{large}, i.e., we then have:
\[
  T \Let M_G R = S \cupplus L
\]
\begin{example}
Consider the program fragment consisting of the statements:
\[
\x_1 \Let a;\; \x_2\Let f(\x_1,a);\; \x_3\Let g(\x_2,f(a,a))
\]
Then $a$ is a ground right-hand side, and $f(a,a)$ is a ground subterm of a non-ground right-hand side,
i.e., $a\in R$ and $f(a,a)\in G$.
Since the term $f(a,a)$ is also contained in $M_G R$, it is \emph{small}.
\qed
\end{example}
Let $\bar R$ be the set of \emph{minimal} elements in $M_G R$ which are large, i.e., not contained in $S$.
Then by Theorem~\ref{t:unique},
every large term $t$ can be uniquely factored such that $t = t'r$
where $t' \in M_G$ and $r \in \bar R$. We then have for small and large terms:
\[
  S \Let M_G R \cap (R^*\cup G) \quad\text{and}\quad L \Let M_G \bar R
\]
where $R^*$ is the subterm closure of $R$.
For small terms, i.e., for terms in $S$, on the other hand, we cannot hope for unique factorizations.
Since there are finitely many small terms only, we take care of small terms by two means:
\begin{itemize}
\item We restrict the formats $[F_{\x,\cdot}]$ and $[F_{\cdot,\x}]$ from the last section to the case where the occurring
      ground terms are large and introduce dedicated sub-formats $[F_{\x,s}]$ and $[F_{s,\x}]$
      for each small term $s$ in the equalities.
\item For $T$-subsumption, we single out the case of subsumption w.r.t.\ assignments of large terms only and
      treat subsumption w.r.t.\ assignments assigning small terms separately.
\end{itemize}
The set of non-ground terms is again given as $T' \Let M_G \X$.
Thus, we now consider the following formats of two-variable equalities:
\[
\begin{array}{@{}l@{\qquad}lcl@{\qquad}l@{}}
{[}F_{\x,\y}]      & As\x &\doteq& Bt\y     & \text{where $s,t \in M_G$} \\
{[}F_{\cdot,\x}]      & As &\doteq& Bt\x       & \text{where $s \in L$ and $t \in M_G$} \\
{[}F_{s,\x}]    & As &\doteq& Bt\x       & \text{where $s \in S$ and $t \in M_G$} \\
{[}F_{\x,\cdot}]      & At\x &\doteq& Bs       & \text{where $s \in L$ and $t \in M_G$} \\
{[}F_{\x,s}]    & At\x &\doteq& Bs       & \text{where $s \in S$ and $t \in M_G$} \\
\end{array}
\]
In the following, let us call
a substitution $\sigma$ of program variables \emph{small}, if for every program variable $\x$,
$\sigma(\x)$ either equals $\x$ or is a small ground term.
The notions of satisfiability, equivalence and subsumption restricted to the set $T$ can be
inferred by means of the corresponding notions restricted to the set $L$ of large terms only.
We have:
\begin{itemize}
\item
A conjunction $\phi$ of equalities is $T$-satisfiable iff there is a small substitution $\sigma$
such that $\sigma(\phi)$ is $L$-satisfiable.
\item
A conjunction $\phi$ $T$-subsumes an equality $e$, iff for every small substitution $\sigma$,
$\sigma(\phi)$ $L$-subsumes $\sigma(e)$.
\end{itemize}

\noindent
According to this observation, it seems plausible to consider the analogue of Theorem~\ref{t:ressubcompact}
for $L$-subsumption and $L$-compactness only. We obtain:

\begin{theorem}\label{t:gensubcompact}~
\begin{description}
\item[$L$-subsumption.]
For finite sets $E,E'$ of two-variable equalities of the same format
it is decidable whether $\bigwedge E$ $L$-subsumes $\bigwedge E'$ or not.
\item[$L$-compactness.]
Every $L$-satisfiable conjunction of a set $E$ of two-variable equalities of the same format is $L$-subsumed
by a conjunction of a subset of at most three equalities in $E$.
\end{description}
\end{theorem}

\begin{proof}
For equalities of the formats $[F_{\x,\y}],[F_{\x,\cdot}],[F_{\cdot,\x}]$ the proofs are analogous to the
corresponding proofs for Theorem~\ref{t:ressubcompact} where
the set $T$ is replaced with the set $L = M_G\bar R$, i.e., instead of the set $R$ we rely on the set $\bar R$
of unique end marker terms.

Now consider equalities of the format $[F_{s,\x}]$ for a small term $s\in S$.
W.l.o.g.\ let $As \doteq Bt\x$ and $As \doteq Bt'\x$ be two equalities of this format.
If $t \neq t'$, then their conjunction is either contradictory,
or $t\x,t'\x$ have a ground unifier which maps $\x$ to a value from $G$ --- in contradiction
to the assumption that $\x$ takes values from $L$ only.

Therefore, each conjunction of a set $E$ of equalities of the format $[F_{s,\x}]$ either is $L$-equivalent to $\False$
or to a single equality in $E$, and the assertion of the theorem follows.
The same argument also applies for the format $[F_{\x,s}]$.
\end{proof}

\noindent
Given that $L$-subsumption is decidable, at least for equalities of the same format,
and that also $L$-compactness holds,
we define an approximate $T$-subsumption relation $\bigwedge E \impliesSharp\bigwedge E'$ as follows.
Let $E_F$ and $E_F'$ denote the subsets of equalities of format $F$, in $E$ and $E'$, respectively.
Then $\bigwedge E \impliesSharp\bigwedge E'$ holds iff
for all small substitutions $\sigma$,
$\bigwedge \sigma(E_F)$ $L$-subsumes $\bigwedge \sigma(E_F')$ for all formats $F$.
As a consequence of Theorem~\ref{t:gensubcompact}, we obtain:

\begin{theorem}\label{t:genapprsubcompact}
Assume that $n$ is the number of program variables and $m$ is the cardinality of the set $S$ of small terms.
\begin{description}
\item[Approximate $T$-subsumption.]
For finite sets $E,E'$ of two-variable equalities,
it is decidable whether $\bigwedge E$ approximately $T$-subsumes $\bigwedge E'$ or not.
\item[Approximate $T$-compactness.]
Every $T$-satisfiable conjunction of a set $E$ of two-variable equalities
is approximately $T$-subsumed by a conjunction of a subset of at most $\mathcal{O}(n^2 \cdot m^2)$ equalities in $E$.
\end{description}
\end{theorem}
\begin{proof}
In the following we consider equalities of formats which contain either one or two program variables.
\begin{description}
\item[One program variable.]
Let $E'$ denote a subset of equalities of $E$ of the same format which contains only the program variable $\x$.
Then for every $c \in S$ we construct a subset $E'_c \subseteq E'$
such that $\bigwedge E'_c[c/\x]$ $T$-subsumes $\bigwedge E'[c/\x]$. Furthermore, we construct a subset $E'_L \subseteq E'$ which $L$-subsumes $E'$.
Then the conjunction of $\bigcup_{c \in S} E'_c \cup E'_L$
$T$-subsumes the conjunction of $E'$.

For each set $E'_c$ we require at most two equalities (according to Corollary~\ref{c:compact})
while for the set $E'_L$ we require at most three equalities
(according to Theorem~\ref{t:gensubcompact}).
Thus, overall, at most $2m+3$ equalities are required.
\item[Two program variables.]
Let $E'$ denote a subset of equalities of $E$ of format $[F_{\x,\y}]$ which contains only the distinct program variables $\x,\y$.
We proceed as follows.
\begin{enumerate}
\item
For every $c\in S$, we construct a set $E'_{c,\y} \subseteq E'$
such that $\bigwedge E'_{c,\y}[c/\x]$ $T$-subsumes $\bigwedge E'[c/\x]$.
\item
For every $c\in S$, we construct a set $E'_{\x,c} \subseteq E'$
such that $\bigwedge E'_{\x,c}[c/\y]$ $T$-subsumes $\bigwedge E'[c/\y]$.
\item
Finally, we construct a set $E'_L \subseteq E'$ such that
$\bigwedge E'_L$ $L$-subsumes $\bigwedge E'$.
\end{enumerate}
Then the conjunction of
$\bigcup_{c\in S} E'_{\x,c} \cup E'_{c,\y} \cup E'_L$
$T$-subsumes the conjunction of $E'$.

For each set $E'_{\x,c}$ resp. $E'_{c,\y}$ we require at most $2m+3$ equalities.
While for the set $E'_L$ we require at most three equalities (according to Theorem~\ref{t:gensubcompact}).
Thus, overall, at most $4m^2 + 6m + 3$ equalities are required for $E'$.
\end{description}
For each program variable $\x$ we distinguish between $2m+3$ different formats ($[F_{\x,s}]$, $[F_{s,\x}]$, $s\in S$,
and $[F_{\x,\x}]$,$[F_{\x,\cdot}]$, and $[F_{\cdot,\x}]$) of equalities.
While for two distinct program variables we only have one format $[F_{\x,\y}]$ of equalities.
Hence we conclude that every conjunction $E$ is $T$-subsumed by a conjunction
of a subset of $E$ which contains at most
\[
n \cdot (2m+3) \cdot (2m+3) + n \cdot (n-1) \cdot (4m^2 + 6m + 3) \hfill\in\hfill \mathcal{O}(n^2 \cdot m^2)
\]
equalities.
This completes the proof.
\end{proof}
 
Due to Theorem~\ref{t:genapprsubcompact}, representations of the greatest solutions
of the constraint systems {\bf S} and {\bf R} can be effectively computed.
By that, we arrive at our main result:
\begin{theorem}\label{t:general}
  Assume that all right-hand sides of assignments of a program contain at most one variable.
  Then all valid inter-procedurally two-variable Herbrand equalities can be inferred.
\end{theorem}

\noindent
The proof is analogous to the proof of Theorem~\ref{t:restricted} --- only that Theorem~\ref{t:genapprsubcompact}
is used instead of Corollary~\ref{c:ressubcompact}.

\begin{example}\label{e:general}
Consider a variant of the program from Section~\ref{s:programs}
where the non-ground assignments are given by:
\[
\x \Let f(\x,a,\x)\quad\text{and}\quad
\y \Let f(\y,a,\y)
\]
The set of small terms then is given by $S=\{a\}$, while the set of smallest
large terms is given by $\bar R=\{f(a,a,a)\}$.

Now consider the constraint system {\bf R} for the recursive procedure $p$
as defined by the control flow graph of Figure~\ref{f:prog} with the modified assignments.
Let us concentrate on the start point $4$ of $p$.
Round-Robin iteration for the transformer $\semT{4}$ for the generic post-condition $A\x\doteq B\y$,
successively will produce the equalities depicted by Table~\ref{t:RRiter2},
\begin{table}
\caption{\label{t:RRiter2}Round-Robin iteration of Example~\ref{e:general}}
\centerline{$\begin{array}{|l||lll|lll|lll|}
\hline
        & &1& \  & & 2 & \ & & 3 & \    \\
\hline\hline
7       & A\x&\doteq &B\y       &&&             \       &&&     \ \\
\hline
6       & A\x&\doteq &B f(\y,a,\y)&&& \       &&&     \ \\
\hline
5       & &\True& \     &               A\x&\doteq &B f(\y,a,\y)&
                                        Af(\x,a,\x)&\doteq &B f(f(\y,a,\y),a,f(\y,a,\y)) \\
\hline
4       & A\x&\doteq &B\y &             Af(\x,a,\x)&\doteq &B f(\y,a,\y)&
                                        Af(f(\x,a,\x),a,f(\x,a,\x))&\doteq &
                                        B f(f(\y,a,\y),a,f(\y,a,\y))\\
\hline
\end{array}$}
\end{table}
where in the $i$th column, we again only have displayed pre-conditions
which have additionally been attained in the $i$th iteration for the program points
$7,6,5$ and $4$, respectively.
For program point $4$, we can argue as in Example~\ref{e:prog2}
in order to verify that the first two equalities $L$-subsume the third one.
Therefore, it remains to consider the given iteration for any small assignment
to the program variables $\x,\y$.

If $\x=\y=a$, then $A=B$ must hold and the third equality is implied.
If $\x=a$, but $\y$ is bound to large terms, then the first equality is of the format
$[F_{a,\y}]$ while the subsequent equalities are of the format $[F_{\cdot,\y}]$.
Accordingly, the first equality must be kept separately.
For the second and third equalities the techniques from Theorem~\ref{t:gensubcompact}
again allow to derive the monoidal equality:
\[
Af(\bullet,a,\bullet)A^{-1} \doteq Bf(\bullet,a,\bullet)B^{-1}
\]
implying that the equality provided in the fourth iteration will be subsumed.
A similar argument applies to the case where $\y=a$ while $\x$ is bound to large values only.
Thus, Round-Robin fixpoint iteration reaches the greatest fixpoint after the fourth iteration.
\qed
\end{example}

\section{Multi-variable Equalities}\label{s:addon}

In this section, we extend our methods to arbitrary equalities such as
\[
\x \doteq f(g\y,\z)
\]
where, w.l.o.g., the left-hand side is a plain program variable while the right-hand side
is a term possibly containing occurrences of more than one variable.
Still, we consider programs where each right-hand side of an assignment contains occurrences
of at most one variable only.
Here, we indicate how for any program point $v$ and any given candidate Herbrand equality $\x \doteq t$, we
verify whether or not the equality is valid whenever $v$ is reached.
There are only constantly many candidate equalities of this form, namely, all equalities which
hold for a variable assignment $\sigma_v$ computed by a single run of the program reaching $v$.
Since such a single run can be effectively computed before-hand, we conclude:
\begin{theorem}
  Assume that all right-hand sides of assignments of a program contain at most one variable.
  Then all inter-procedurally valid Herbrand equalities can be inferred.
\end{theorem}

Now consider the single Herbrand equality $\x \doteq t$,
where $t$ contains occurrences of the program variables
$\y_1,\dotsc,\y_k$.
Then we construct new generic post-conditions as follows. First, we consider all substitutions $\sigma$
which map each variable $\y_i$ in $t$ either to a fresh template variable $C_i$ or an expression
$A_i\y'_i$ for a fresh template variable $A_i$ and any program variable $\y'_i$.
Then the new generic post-conditions are of the form $\x' \doteq t'$ where $\x'$ is any program variable,
and $t'$ is a subterm of $t\sigma$.
Note that this set may be large but is still finite.
In a practical implementation, we may, however, tabulate for each procedure the weakest pre-conditions
only for those post-conditions which are really required.
Since we envision that for realistic programs, only few of these equalities for each procedure
will be necessary to prove the queried assertion $e$ at target point $u$, the potential exponential blow-up
will still be not an obstacle.

\begin{example}
Assume the equality we are interested in is $\x \doteq f(g\y,\z)$, then, e.g.,
\[
\x\doteq f (gA_1\y, A_2\z)\qquad
\y\doteq f (gA_1\x, A_2\z)
\]
are new generic post-conditions to be considered, as well as
\begin{flalign*}
&& \z\doteq f(g C, A\y) \qquad \y\doteq f(g A\z, C) &&\null\qEd
\end{flalign*}
\end{example}

Starting from a new generic post-condition $\x \doteq p$, repeatedly computing weakest pre-conditions w.r.t.\ assignments
may result in conjunctions of equalities which can be simplified to one of the following forms:
\begin{itemize}
\item $s \doteq C_i$ or $s\doteq A_i t_i$ where $s$ and $t_i$ contain occurrences of at most one program variable each;
\item $\y \doteq p'$, i.e., the left-hand side is a plain program variable, and the right-hand side $p'$ is obtained
      from a subterm of $p$ by substituting each occurrence of a program variable
      $\y_i$ with some term $t_i$ containing occurrences of at most one program variable each.
\end{itemize}
\begin{example}
  Consider, e.g., the generic post-condition $\x \doteq f(gA_1\y,A_2\z)$. Then
  \[
  \begin{array}{l@{\quad}l@{\quad}l}
  \semT{\x \Let f(\x,h\x)}(\x \doteq f(gA_1\y,A_2\z))	&=&f(\x,h\x)\doteq f(gA_1\y,A_2\z) \\
    &=&(\x\doteq gA_1\y)\land(h\x\doteq A_2\z)
  \end{array}
  \]
  which means that we equivalently obtain two two-variable equalities.
  Likewise, for an assignment to one of the program variables on the right, we
  have:
  \[
  \begin{array}{l@{\quad}l@{\quad}l}
    \semT{\y \Let f(b,\y)}(\x \doteq f(gA_1\y,A_2\z)) &=&\x\doteq f(gA_1f(b,\y),A_2\z) \\
  \end{array}
  \]
  which is an equality of the form described in the second item.
  \qed
\end{example}
The equalities from the first item contain at most one program variable on each side.
They can be dealt with in the same way as we did for plain two-variable equalities.
They are even somewhat simpler, in that only one template variable occurs (instead of two).
The equalities of the second item, on the other hand, we may group into equalities which agree
in the variable on the left as well as in the constructor applications outside the template variables $A_i$.
Of each such group it suffices to keep exactly one equality.
Any conjunction with another equality from the same group will allow us to
simplify the second equality to a conjunction of equalities with at most one
program variable on each side.

\begin{example}
  Assume that we are given the conjunction of the two equalities:
  \[
  \x \doteq f(gA_1\y,A_2\z) \qquad \x \doteq f(gA_3h\y,A_4g\z)
  \]
  This conjunction is equivalent to the first equality together with:
  \[
  f(gA_1\y,A_2\z) \doteq f(gA_3h\y,A_4g\z)
  \]
  The latter equality, now, is equivalent to the conjunction of:
  \[
  A_1\y\doteq A_3h\y \qquad A_2\z\doteq A_4g\z
  \]
  which is a finite conjunction of two-variable equalities.
  \qed
\end{example}

Thus, in the course of {\bf WP} computation for any of the new generic post-conditions,
we obtain conjunctions which (up to finitely many exceptions) consists of
two-variable equalities only,
to which we can apply our methods from Section~\ref{s:general}.
In summary, we thus find that it can be effectively verified whether or not a general Herbrand equality
is inter-procedurally valid at a given program point $v$.

\section{Analysis of Computational Complexity}\label{s:polytime}

In the following we indicate how our algorithms for inferring inter-procedurally
valid Herbrand equalities can be realized in polynomial time.
Crucial for the complexity is the size of representations of occurring terms.
Note that already the factorization of a term results in a succinct representation by sharing isomorphic subtrees.
Still, the \emph{depth} of occurring terms may grow exponentially in a program with procedures.
\begin{example}
  Consider the following program fragment consisting of procedures $p_n$ and two global variables $\x$ and $\y$:
  \[\begin{array}{ll}
  p_i           & \{\; p_{i-1}();\; p_{i-1}(); \;\} \\[.5ex]
  p_0           & \{\; \x \Let f(\x,\x);\; \y \Let f(\y,\y); \;\}
  \end{array}\]
  The weakest pre-condition of a generic post-condition $A\x \doteq B\y$ for a procedure $p_n$ is then given by a single
  equality $Af(\bullet,\bullet)^{2^n}\x \doteq Bf(\bullet,\bullet)^{2^n}\y$ with exponentially deep terms on both sides of the equality.
  \qed
\end{example}
Hence, in order to arrive at polynomial algorithms,
polynomially sized representations must be provided for all occurring terms which additionally support the
required operations on terms in polynomial time.
For trees, \emph{tree straight-line programs} (TSLP, for short) have been proposed
which efficiently represent trees by context-free \emph{tree} grammars~(see \cite{Schauss13,Lohrey15} for recent overviews). 
Polynomial algorithms for equality of the represented trees, however, are only known in case that the tree grammars
in question are \emph{linear} --- meaning that each parameter of a rule occurs in the corresponding right-hand side at most once.
Our factorizations of trees, however, may easily introduce \emph{non-linear} terms.
Therefore, we apply compression only to elements from the free monoid $M_G$. 
We use ordinary straight-line programs (SLP for short) --- but
with the understanding that individual letters are irreducible trees.
For plain symbols (corresponding to unary constructors only), 
algorithms based on such a representation have been sketched in~\cite{Gulwani07}.
Thus in our application, an SLP $P$ of \emph{size} $k$ consists of a sequence of definitions
\[
X_i \SLPassign \alpha_i \qquad i=1,\dotsc,k
\]
where either $k=1$ and $\alpha_i=\bullet$, or each right-hand side $\alpha_i$ is either of
the form $X_j X_l$ for unknowns $X_j, X_l$ with $i<j,l$ or a single irreducible term $t\in M_G$.
Given a suitable ordering on the unknowns together with an initial unknown, we may consider $P$ also as a
\emph{set} of definitions of unknowns.
Beyond the size, we are also interested in the \emph{depth}, i.e., the length $h$ of the longest chain of unknowns
$Y_1,\ldots,Y_{h}$ in $P$
such that $Y_1\SLPassign\alpha_1 Y_2\alpha'_1,\dotsc,Y_{h-1}\SLPassign\alpha_{h-1}Y_h\alpha'_{h-1}$ occur among the definitions in $P$
for suitable $\alpha_i,\alpha'_i$.
An SLP can also be considered as a context-free grammar (in Chomsky Normal Form) generating a single term in $M_G$. Formally,
the term $\sem{P}$ represented by $P$ is defined by $\sem{P} = \sem{X_1}_P$ where
\begin{align*}
  \sem{X_i}_P &= \sem{X_j}_P \sem{X_l}_P && (X_i \SLPassign X_j X_l) \in P \\
  \sem{X_i}_P &= t && \text{$(X_i \SLPassign t) \in P$ and $t \in M_G$} \\
  \intertext{We remark that in linear time in the size of $P$, we can determine the \emph{length}
  of the represented element in $M_G$, which is defined by:}
  \length{X_i}_P &= \length{X_j}_P + \length{X_l}_P && (X_i \SLPassign X_j X_l) \in P \\
  \length{X_i}_P &= 0 && \text{$(X_i \SLPassign \bullet) \in P$} \\
  \length{X_i}_P &= 1 && \text{$(X_i \SLPassign t) \in P$ and $t \in M_G\setminus\{\bullet\}$}
\end{align*}
An SLP in Chomsky normal form of size $k$ cannot produce a word
larger than $2^k$. Therefore, the length of each word which it generates can be described
by $k$ bits. For such numbers, basic operations
as equality and addition can be done in linear time in $k$.

In order to avoid repeated computation of lengths, we assume that every unknown occurring during the
analysis will once for all be annotated with its length.
For later use, we collect a set of basic algorithms for SLPs (see, e.g.,~\cite{Lohrey12}).
\begin{theorem}\label{t:termops}
The following tasks can be realized in polynomial time:
\begin{enumerate}
\item\label{termop1}
  Given an SLP $P$ representing a term $t\in M_G$.
  Determine an SLP $Q$ for the reverse of $t$ such that $Q$ has the same size and depth as $P$.
\item\label{termop2}
  Given an SLP $P$ representing a term $t\in M_G$ of some length $k$, and some number $0\leq h\leq k$.
  Determine an SLP $Q$ for the prefix (suffix) of $t$ of length $h$. The number of new
  definitions in $Q$ is bounded by the depth of $P$, and the depth of $Q$ is not increased.
\item\label{termop3}
  Given SLPs $P$ and $Q$ for terms $t,t'\in M_G$. Determine whether or not $t=t'$.
\item\label{termop4}
  Given SLPs $P$ and $Q$ for terms $t,t'\in M_G$. Determine the length of the longest common
  prefix (suffix) of $t,t'$.
\item\label{termop5}
  Given SLPs $P$ and $Q$ for terms $t,t'\in M_G$. Determine an SLP for $tt'$. At most one
  new definition is introduced and also the depth is increased at most by one.
\end{enumerate}
\end{theorem}
\begin{proof}
An SLP for the reverse of $t$ is obtained from $P$ by introducing a fresh copy of unknowns $X'$ for every unknown $X$ in $P$
together with a definition $X'\SLPassign f$ if $X\SLPassign f$ with $f\in M_G$, and a definition $X'\SLPassign Z'Y'$ if $P$ has a definition
$X\SLPassign YZ$. This new SLP clearly generates the reverse of the SLP $P$ --- proving assertion~\ref{termop1}.

For a proof of assertion~\ref{termop2}, we only consider the construction of an SLP for the prefix of $t$ of length $h$.
The case where $h=0$ is trivial. Therefore, assume that $h>0$.
We construct the new SLP by successively introducing fresh unknowns $X'$ for the unknowns $X$ on a path in $P$.
in order to do so, we maintain the sum of the lengths $l$
of the unknowns to the left of the path.
We start with the initial unknown $X_1$ of $P$ where $l=0$ with corresponding fresh unknown $X'_1$.
In general, assume that $l < h$, and we have reached an unknown $X$ with corresponding fresh unknown $X'$.
First assume that the
definition of $X$ in $P$ is given by $X\SLPassign f$ for some irreducible term $f\in M_G$.
In this case, $h = l+1$, and we set the definition of $X'$ to $X'\SLPassign f$.
Then assume that the definition of $X$ in $P$ is given by $X\SLPassign Y Z$.
If $h\leq l+\length{Y}_P$, then we introduce a fresh copy $Y'$ for $Y$ and the definition $X'\SLPassign Y'$ for $X'$,
and proceed with $Y'$.
If $l+\length{Y}_P < h$, then we introduce a fresh copy $Z'$ for $Z$ and the definition $X'\SLPassign Y Z'$ for $X'$
and proceed with $Z'$.
The resulting set of definitions, though, may not meet our assumptions on SLPs. The definitions with single unknowns in
their right-hand sides, can however, be removed in polynomial time by a technique similar to the removal of chain rules
in context-free grammars.

Polynomial time algorithms for deciding equivalence of SLPs were independently discovered by
Hirshfeld et al.~\cite{Hirshfeld96}, Mehlhorn et al.~\cite{Mehlhorn97}, and Plandowski~\cite{Plandowski94}
proving assertion~\ref{termop3}.
The algorithms can be applied to obtain a polynomial time algorithm for determining the length of longest common
prefixes of elements in a free monoid as claimed in assertion~\ref{termop4}. First, the algorithm from assertion~\ref{termop3} can be extended to decide
whether or not $t$ is a \emph{prefix} of $t'$ by first determining the lengths $h$ and $h'$ of $t$ and $t'$, respectively. If
$h > h'$, $t$ is not a prefix of $t'$.
Otherwise, we may determine an SLP $Q'$ of $Q$ representing the prefix of $t'$ of length $h$ which then is checked for
equivalence with $P$.
In the next step, that algorithm is extended to the case where $t$ is not necessarily a prefix of $t'$
by performing binary search on the prefixes of $t$.

Finally, consider assertion~\ref{termop5}.
If $t$ or $t'$ equals $\bullet$, the concatenation is trivial.
So assume that neither $t$ nor $t'$ equal $\bullet$, and that the
initial unknowns of the SLPs $P$ and $Q$ equal $X_1$ and $Y_1$, respectively.
Let $X_0$ denote a fresh unknown. Then
the term $tt'$ can be represented by the SLP $P\cup Q$ together with the initial definition $X_0\SLPassign X_1 Y_1$.
\end{proof}

The size of a term $t \in \T_\Omega(\X) \cup \T_\Omega(\bullet)$ is given by $\sizeT{t}$ which is recursively defined as follows:
\[\begin{array}{lcl@{\qquad}l}
    \sizeT{t} &=& 1 + \Sigma_{i=1}^k \sizeT{t_i} & \text{if $t = f(t_1,\dotsc,t_k)$ and $f \in \Omega_k$} \\
    \sizeT{t} &=& 1 & \text{if $t \in \X \cup \{\bullet\}$}
\end{array}\]
In the following we define the size of a program. As mentioned in Section~\ref{s:programs} we
do not operate on the syntax of a program directly but on the corresponding control flow graph.
The size of a program is then given as the sum of the number of nodes, the number of edges, and
the sum of the sizes of terms of right-hand sides of assignments.

A non-ground term $t=t'\x$ containing occurrences of the variable $\x$ is then succinctly represented by the pair $(P,\x)$ where
$P$ is an SLP for $t'$.
Ground terms in $T$ may be factorized differently for initialization-restricted or unrestricted programs.
In the following, we first consider initialization-restricted programs, and subsequently unrestricted programs.

\subsection{Polynomial-time Algorithms for IR Programs}\label{ss:ir_ptime}

For \emph{initialization-restricted} programs,
every ground term $t$ possibly produced at run-time, can be uniquely
factored into $t = t'r$ for $t'\in M_G$ and a ground term $r \in R$ occurring as a right-hand side in the program.
Such a term $t$ is represented by a pair $(P,r)$ where $P$ is an SLP for $t'$.
We remark that the size of the term $r$ is bounded by the size of the program.

In a succinct representation of a post-condition $\phi$, every occurring term in $T \cup T'$
(recall that $T = M_G R$ and $T' = M_G \X$)
is represented by such a pair where the different SLPs need not necessarily be disjoint but may share
unknowns together with their definitions.
The weakest pre-condition of a post-condition $\phi$ w.r.t.\ a non-ground assignment $\x \Let t\y$ is given as $\phi[t\y/\x]$.
This means that $t\y$ must be substituted into each term $s\x$, $s \in M_G$ occurring in $\phi$.
If $s$ or $t$ equals $\bullet$, the substitution is trivial.
So assume that neither $s$ nor $t$ equal $\bullet$.
Then by Theorem~\ref{t:termops} an SLP $P$ for $st$ can be constructed from the SLPs for $s$ and $t$
by adding one fresh unknown together with its definition, so that
the depth of the involved SLPs increases at most by one --- even if the depth of the resulting term may be doubled.
The resulting term of the substitution is then represented by the pair $(P,\y)$.

Now consider a substitution $\phi[t/\x]$ for a ground term $t = t'r$ where $t'\in M_G$ and
$r \in R$ is a ground term of some assignment.
This means that $t$ must be substituted into each term $s\x$ occurring in $\phi$.
If $s$ equals $\bullet$, the substitution is trivial.
Therefore, assume that $s$ does not equal $\bullet$.
Then by Theorem~\ref{t:termops} an SLP $P$ for $st'$ can be determined from the SLPs for $s$ and $t'$
in polynomial time. The resulting term of the substitution is then represented by the pair $(P,r)$.
We thus have proven:
\begin{lemma}\label{l:ir-succ_subst}
  Consider a single equality $As_1\doteq Bs_2$ or $As\doteq C$ where $s_1,s_2,s \in T \cup T'$ are succinctly represented.
  Then a succinct representation of the weakest pre-condition of the equality w.r.t.\ an assignment
  $\x \Let t$ can be determined in time polynomial in the size of $t$.
  \qed
\end{lemma}
The weakest pre-condition of a post-condition $As\x \doteq Bt\y$ w.r.t.\ a procedure call $p()$ is given as
$\phi'= \phi[As/A,Bt/B]$ if the weakest pre-condition of the generic post-condition $A\x \doteq B\y$ w.r.t.\ a procedure call $p()$
is given as $\phi$. This case is similar to the case of \mbox{(non-)}ground program variable assignments. That means that, instead
of a program variable two template variables are substituted.
In order to obtain succinct representations for the terms in $\phi'$, we again can apply our techniques for computing succinct
representations for the result of the substitution of terms.
\begin{lemma}\label{l:ir-succ_adapt}
  Consider a single equality $As\x \doteq Bt\y$ (resp.\ $As \doteq Bt\x$, $As\x \doteq Bt$, or $As\x \doteq C$)
  where the occurring terms $s,t \in M_G$ are succinctly represented.
  Moreover, assume that each term of type $T \cup T'$ occurring in the weakest pre-condition $\phi$ of a generic post-condition
  $A\x \doteq B\y$ (resp.\ $A\x \doteq C$) w.r.t.\ a procedure call $p()$ is also succinctly represented.
  Then a succinct representation of the weakest pre-condition of the equality w.r.t.\ a procedure call $p()$
  can be computed in time polynomial in the number of equalities in $\phi$.
  \qed
\end{lemma}
From Lemmas~\ref{l:ir-succ_subst} and~\ref{l:ir-succ_adapt}, we conclude that the sizes and depths of occurring SLPs
during the whole fixpoint computation for determining the {\bf WP} transformers for procedures as well as the
{\bf WP} transformers for reachability, remains polynomial in
the size of the program and the numbers of equalities occurring in pre-conditions.
Accordingly, a polynomial time algorithm for inferring valid Herbrand equalities is obtained whenever we
are given polynomial time algorithms for
\begin{itemize}
\item solving systems of ground equalities, as well as for
\item approximate $T$-subsumption.
\end{itemize}

Consider a satisfiable equality of the form $As \doteq Bt$ where $s,t \in T$ are ground.
Let $A=\bullet$, then the finite set of all solutions for $B$ equals the set
\[
  \Set{uw \in \C_\Omega | \text{$s = uvt$ and $u,v \in M_G$ and $v$ is irreducable and $wt = vt$}}.
\]
In the set above, each $w$ equals $v$ where some occurrences of $\bullet$ are substituted by $t$.
That means, once the decomposition of $s$ into $uvt$ is known, then all solutions can be trivially derived.
Still there exist $2^i - 1$ many solutions if $\bullet$ occurs $i$ times in the term $v$.
Let $u = \sem{P}$ be represented by some SLP $P$ and $t = \sem{Q}r$ be represented by some SLP $Q$ and $r \in R$.
Then the set of all solutions for $B$ is \emph{succinctly represented} by the tuple
\begin{equation}\label{e:succallsol}
\langle P,v,Q,r \rangle
\end{equation}
Similarly, the finite set of all solutions for the template variable $A$ is succinctly represented by a tuple of the form~\eqref{e:succallsol}, if $B = \bullet$.
\begin{theorem}\label{t:ir-succgroundsolve}
  In the following consider only equalities of the form $As \doteq Bt$ where $s,t \in T$ are ground and succinctly represented.
  \begin{enumerate}
  \item
    It is decidable in polynomial time whether or not the equality $As \doteq Bt$ is satisfiable where $A$ or $B$ receives the value $\bullet$.
    Furthermore, if it is satisfiable,
    then a succinct representation of the form~\eqref{e:succallsol} of the set of all solutions for $A$ (resp.\ $B$) can be determined in polynomial time.
  \item
    It is decidable in polynomial time whether or not the conjunction of the two distinct equalities
    $As_1 \doteq Bt_1$ and $As_2 \doteq Bt_2$ is satisfiable where $A$ or $B$ receives the value $\bullet$.
    Furthermore, if it is satisfiable,
    then a succinct representation of the unique solution can be determined in polynomial time.
  \end{enumerate}
\end{theorem}
\begin{proof}~
  \begin{enumerate}
    \item
      Let $A=\bullet$, i.e., we then consider $s \doteq Bt$.
      If the equality is satisfiable, then $s=t't$ for some $t' \in M_G$ must hold.
      Whether or not $t$ is a suffix of $s$ is decidable in polynomial time.

      Assume that the equality is satisfiable.
      Then each solution of $B$ equals $s$ where some occurrences of $t$ are substituted by $\bullet$.
      Let $s = u v t$ for some $u,v \in M_G$ and $v$ is an irreducible element in $M_G$.
      A succinct representation $Q$ of the prefix $u$ of $s$ of length $\length{s} - \length{t} - 1$
      can be determined in polynomial time.
      Likewise, the irreducible element $v$ occurring in the unique factorization of $s$
      can be determined in polynomial time. Assume that $t$ is succinctly represented by the tuple $(P,r)$.
      Then the set of all solutions for $B$ is succinctly represented by the tuple $\langle Q,v,P,r \rangle$
      of the form~\eqref{e:succallsol}, from which the assertion of this part follows.
    \item
      Let $A=\bullet$, i.e., we then consider $s_1 \doteq Bt_1$ and $s_2 \doteq Bt_2$.
      If the conjunction of the two equalities is satisfiable, then $s_1 = t t_1$ and
      $s_2 = t t_2$ for some $t \in M_G$ must hold, i.e., $B = t$ is then a solution.
      From the succinctly represented term $s_i$ a succinct representation of the prefix $u_i$
      of length $\length{s_i} - \length{t_i}$
      and the suffix $v_i$ of length $\length{t_i}$ can be determined in polynomial time for $i=1,2$.
      If $u_1 = u_2$ and $v_1 = t_1$ and $v_2 = t_2$ holds, then the conjunction is satisfiable
      and $u_1$ is a solution for $B$. This is decidable in polynomial time.

      According to Theorem~\ref{t:solve}, $t$ is a unique solution, i.e., there exists no other
      solution $t' \neq t$. A similar argument holds for the case $B=\bullet$.
      \qedhere
  \end{enumerate}
\end{proof}
Assume that we are given a conjunction of ground equalities arising from the analysis.
Clearly, it allows to efficiently \emph{test} any candidate templates whether or not they constitute a solution.
In light of Theorem~\ref{t:ir-succgroundsolve}, the conjunction allows to \emph{infer}
a succinct representation of all valid equalities in polynomial time.
\begin{theorem}\label{t:ir-subonevarptime}
  $T$-subsumption for equalities of the form $As \doteq C$ where $s \in T \cup T'$
  are succinctly represented is decidable in polynomial time.
\end{theorem}
\begin{proof}
  Consider two equalities $As\x \doteq C$ and $At\x \doteq C$ with $s,t \in M_G$
  (resp. $As \doteq C$ and $At \doteq C$ with $s,t \in T$).
  The conjunction of the two equalities is $T$-unsatisfiable, if $s \neq t$ holds which is decidable in polynomial time.
  Otherwise, if $s = t$ holds, then one equality is subsumed by the other.
\end{proof}

In the following we show that approximate $T$-subsumption of two-variable equalities
is decidable in polynomial time, too. In order to do so we first extend the idea of
succinctly represented terms in $M_G$ to terms in the corresponding free group $F_G$.
That means that definitions of an SLP representing a term in $F_G$ are now either of the form
$X \SLPassign YZ$ for suitable unknowns $Y,Z$ or $X \SLPassign f$ where $f$ is an irreducible term in $F_G$.
The length $\length{t}$ of a term $t \in F_G$ can be determined in time linear
in the size of the SLP representing $t$ similar to any term $s \in M_G$.
The balance $\balance{t}$ of a term $t \in F_G$ which is represented by the SLP $P$
can be determined in linear time in the size of $P$ as follows:
\[\begin{array}{l@{\qquad}l}
    \balance{X_i}_P = \balance{X_j}_P + \balance{X_l}_P & (X_i \SLPassign X_j X_l) \in P \\
    \balance{X_i}_P = 0 & (X_i \SLPassign \bullet) \in P \\
    \balance{X_i}_P = 1 & (X_i \SLPassign f) \in P \text{ and } f \in M_G \setminus \{\bullet\} \\
    \balance{X_i}_P = -1 & (X_i \SLPassign f^-) \in P \text{ and } f \in M_G \setminus \{\bullet\}
\end{array}\]
An SLP in Chomsky normal form of size $k$ cannot produce a word
larger than $2^k$. Therefore, the balance of each word which it generates can be described
by $k+1$ bits. For such numbers, basic operations
as equality, addition and subtraction can be done in time linear in $k$ --- even if only single bit operations
are considered as constant time.
\begin{lemma}\label{l:grouptermops}
  Assume that all terms are succinctly represented and let $F_G$ be the corresponding
  free group of $M_G$. Then the following tasks can be realized in polynomial time:
  \begin{enumerate}
  \item\label{grouptermop0} All tasks described in Theorem~\ref{t:termops} can also be realized for terms in $F_G$.
  \item\label{grouptermop1} Given a term $w \in F_G$, determine the term $w^{-1} \in F_G$.
  \item\label{grouptermop2} Given two maximally canceled terms $u,v \in F_G$, determine $w = uv$ such that $w$ is maximally canceled.
  \item\label{grouptermop3} Given a term $w \in F_G$, determine the term $w^r$, $r \geq 1$.
  \end{enumerate}
\end{lemma}
\begin{proof}
  For the tasks described in Theorem~\ref{t:termops} it is irrelevant from which algebraic structure
  an element $f$ in a definition $X \SLPassign f$ comes. That means, it does not matter if $f \in M_G$
  or $f \in F_G$ proving assertion~\ref{grouptermop0}.

  Given an SLP $P$ representing some term $w \in F_G$, the SLP $P'$ representing the
  term $w^{-1}$ can be constructed as follows. If the definition $X \SLPassign YZ$ is included
  in $P$, then let $X' \SLPassign Z'Y'$ be included in $P'$. Otherwise, if the definition
  $X \SLPassign f$, $f \in F_G$ is included in $P$, then let $X' \SLPassign f^{-1}$ be in $P'$.
  The size and the depth of $P$ and $P'$ are the same proving assertion~\ref{grouptermop1}.

  Assume that the SLPs $P$ and $Q$ represent the terms $u$ and $v$ from $F_G$, respectively.
  By assertion~\ref{grouptermop1}, an SLP for $u^{-1}$ can be determined in polynomial time.
  Furthermore, by Theorem~\ref{t:termops}, the length $k$ of the longest
  common prefix of $u^{-1}$ and $v$ can be determined in polynomial time.
  Again by Theorem~\ref{t:termops}, an SLP $Q$ for the prefix $u'$ of $u$ of length $\length{u} - k$
  can be determined in polynomial time. Similarly, an SLP $Q'$ for the suffix $v'$ of $v$ of length $\length{v} - k$
  can be determined in polynomial time.
  Finally, an SLP for the term $w=u'v'$ can be determined in polynomial time. Since $w$ is maximally canceled, this
  proves assertion~\ref{grouptermop2}.

  The last assertion~\ref{grouptermop3} can be proven as follows.
  The case where $r=1$ is trivial. Therefore, assume that $r > 1$.
  Let the term $w$ be represented by the SLP $P$ with initial unknown $X_0$ and size $s_P$.
  The term $w^{2^k}$, $k \geq 1$ is then represented by the SLP $Q_k$ with initial unknown $N_{k+1}$ and
  the following definitions (for fresh unknowns $N_k$):
  \[\begin{array}{p{\widthof{$N_{i+1}$}}cp{1.5cm}@{\qquad}p{2.8cm}}
    $N_1$ &\SLPassign& $X_0X_0$ \\
    $N_{i+1}$ &\SLPassign& $N_iN_i$ & $1 \leq i \leq \log_2(k)$
  \end{array}\]
  Assume that the binary representation of $r$ equals $b_{\log_2(r)}\ldots b_0$
  where $b_0$ is the least significant bit and
  let $j_1<\dotsb<j_n$ equal the list of indices $j$ where $b_j=1$.
  Then we introduce the SLP $Q$ with initial unknown $M_{1}$ and the following fresh definitions:
  \[\begin{array}{p{\widthof{$N_{j+1}$}}cp{1.5cm}@{\qquad}p{2.8cm}}
    $M_{k}$ &\SLPassign& $N_{j_k}M_{{k+1}}$ & \text{for $1 \leq k < n$} \\
    $M_n$ &\SLPassign& $N_{j_n}$
  \end{array}\]
  Thus, the SLP $Q$ represents $w^r$.
  The size of $Q$ is in $\mathcal{O}(\log_2(r) + s_P)$ from which the assertion follows.
\end{proof}
A term $uv$ which is not maximally canceled, may only be constructed during checks of subsumption
when two terms $u,v \in F_G$ are concatenated. According to Lemma~\ref{l:grouptermops}, however,
a maximally canceled term corresponding to $uv$ can be determined in polynomial time.
Therefore, in the following we assume that each succinctly represented term occurring during
subsumption checks are maximally canceled.
\begin{lemma}\label{l:baseptime}
  Assume that all occurring terms are succinctly represented and maximally canceled.
  Then the assertion of Lemma~\ref{l:base} is decidable in polynomial time, i.e.,
  the question whether
  for an equality of the form $AuA^{-1} \doteq Bu'B^{-1}$ with
	$u,u' \in F_G$ and $\balance{u} = \balance{u'} = 0$, it is decidable in polynomial time,
  whether it is trivial, is equivalent to an equality $As \doteq B$ or $A \doteq Bs$
  for some $s \in M_G$, or is contradictory.
\end{lemma}
\begin{proof}
  Assume that $u$ and $u'$  are represented by the SLPs $P$ and $Q$, respectively.
  The equality is trivial iff $\sem{P} = \bullet = \sem{Q}$ which can be checked in
  constant time since we assumed that succinctly represented terms are maximally canceled.
  If $\sem{P} = \bullet \neq \sem{Q}$, or $\sem{P} \neq \bullet = \sem{Q}$ holds,
  then the equality is contradictory. The latter can also be checked in constant time.

  Otherwise, we proceed as follows.
  The length $n \leq \length{u}$ of the longest positive prefix of $u$ can be determined
  similarly to the length $\length{u}$ of $u$, and thus can be determined in time linear in the size of $P$.
  Likewise, the length $m \leq \length{u}$ of the longest negative suffix of $u$ can be determined in polynomial time,
  by first computing the inverse of $u$, i.e., $u^{-1}$ and then determining the
  longest positive prefix of $u^{-1}$.
  We then proceed by determining SLPs for the prefix $x$ of $u$ of length $n$ and
  the remaining suffix $w$ of $u$ of length $\length{u} - n$. From the SLP representing $w$ we then derive
  SLPs for the prefix $y$ of length $\length{w} - m$ and suffix of length $m$ of $w$
  such that $u = xyz^{-1}$. This can be done in polynomial time.

  Similarly, we determine succinct representations for the longest positive prefix $x'$
  of $u'$, longest negative suffix $z'^{-1}$ and $y'$ such that $u' = x'y'z'^{-1}$.

  Overall this means that the equivalent simplified
  conjunction $Ax \doteq Bx' \land y \doteq y' \land Az \doteq Bz'$
  can be determined in polynomial time.
  Since $y,y' \in M_G$, their equality can be checked in polynomial time.
  If the conjunction is satisfiable then it is equivalent
  to a solved equality $As \doteq B$ or $A \doteq Bs$ which means that either $x=sx'$ and $z=sz'$
  or $x'=sx$ and $z'=sz$ holds which can be checked in polynomial time.
\end{proof}

\begin{lemma}\label{l:basetwoptime}
  Assume that all occurring terms are succinctly represented and maximally canceled.
  Then the assertion of Theorem~\ref{t:base} is decidable in polynomial time, i.e.,
  it is decidable in polynomial time whether
  the conjunction of the two equalities $AuA^{-1} \doteq Bu'B^{-1}$ and $AvA^{-1} \doteq Bv'B^{-1}$ with $u,u',v,v' \in F_G$
  is equivalent to one solved equality, or to a single equality, or are contradictory.
\end{lemma}
\begin{proof}
  W.l.o.g.\ assume that $\balance{u} \geq \balance{v}$.
  If $\balance{v} = 0$, then from Lemma~\ref{l:baseptime} follows that
  $AvA^{-1} \doteq Bv'B^{-1}$ is either trivial, i.e., the conjunction of the two initial equalities
  is equivalent to $AuA^{-1} \doteq Bu'B^{-1}$, or is contradictory, i.e., the conjunction of the
  two initial equalities is equivalent to $AvA^{-1} \doteq Bv'B^{-1}$, or the equality is equivalent
  to one solved equality $As \doteq B$ (resp. $A \doteq Bs$). In the latter case either holds
  $u = su's^{-1}$ (resp. $u' = sus^{-1}$) and the conjunction of the two equalities is
  equivalent to $AvA^{-1} \doteq Bv'B^{-1}$ or the conjunction is contradictory.
  According to Theorem~\ref{t:termops} and Lemma~\ref{l:grouptermops} the equality check
  $u = su's^{-1}$ (resp. $u' = sus^{-1}$) can be done in polynomial time ---
  from which the assertion of this part follows.

  Otherwise, if $\balance{v} > 0$, then let $r = \balance{u} \bmod \balance{v}$
  and we derive a third equality $AwA^{-1} \doteq Bw'B^{-1}$
  such that $w = uv^{-r}$ and $w' = u'v'^{-r}$. According to Lemma~\ref{l:grouptermops}
  the terms $w,w'$ can be determined in polynomial time.
  We then start allover by considering the two equalities
  $AvA^{-1} \doteq Bv'B^{-1}$ and $AwA^{-1} \doteq Bw'B^{-1}$ where $\balance{v} \geq \balance{w}$ holds.
  This algorithm is a generalization of \emph{Euclid's algorithm}.
  Since Euclid's algorithm performs at most logarithmic many iterations~\cite[pp. 21--22]{Mollin08}
  and in each iteration we introduce logarithmic many new unknowns,
  the assertion of the theorem follows.
\end{proof}

\begin{theorem}\label{t:ir-tsubptime}
  For finite sets $E,E'$ of equalities of the form $As \doteq Bt$ where $s,t \in T \cup T'$ are succinctly represented,
  it is decidable in polynomial time whether $\bigwedge E$ approximately $T$-subsumes $\bigwedge E'$ or not,
  whenever $A$ or $B$ equals $\bullet$.
\end{theorem}
\begin{proof}
  Consider equalities of the form $As \doteq Bt$ where $s,t \in T$ are ground terms.
  According to Theorem~\ref{t:ir-succgroundsolve} $T$-subsumption is decidable in polynomial time.

  Consider the three equalities $As_i\x \doteq Bt_i\y$, $i=1,2,3$ and let
  w.l.o.g.\ $\balance{s_1} \geq \balance{s_2},\balance{s_3}$. We then derive the two equalities
  $AuA^{-1} \doteq Bu'B^{-1}$ and $AvA^{-1} \doteq Bv'B^{-1}$
  where $u \equiv s_1s_2^{-1}$, $u' \equiv t_1t_2^{-1}$, $v \equiv s_1s_3^{-1}$, and $v' \equiv t_1t_3^{-1}$
  are maximally canceled
  in polynomial time. According to Lemma~\ref{l:basetwoptime} it is decidable in polynomial time
  whether the conjunction is unsatisfiable, or equivalent to one equality, i.e.,
  equality $As_1\x \doteq Bt_1\y$ is then subsumed, or is equivalent to one solved equality.
  In the latter case from a fourth equality either follows the same solved equality and is therefore subsumed
  or is contradictory.
  A similar argument holds for equalities of the format $As \doteq Bt\x$ (resp. $At\x \doteq Bs$).

  We conclude that $T$-subsumption for equalities of the same format is decidable in polynomial time.
  Since we consider only polynomial many different formats of equalities, the assertion of the theorem follows.
\end{proof}

\begin{theorem}
  Assume that all right-hand sides of assignments of an initialization-restricted program contain at most one variable.
  Then for every program point $u$ and program variables $\x$ and $\y$,
  a succinct representation of the form~\eqref{e:succallsol} of the set of all valid two-variable Herbrand equalities between $\x$ and $\y$,
  can be determined in time polynomial in the size of the program.
  \qed
\end{theorem}
 
\subsection{Polynomial-time Algorithms for Unrestricted Programs}\label{ss:unrestricted_ptime}

For unrestricted programs there need not exist a unique factorization for every possible run-time value.
Only for large terms, i.e., terms in $L = M_G \bar R$, unique factorizations are possible.
Accordingly, a large term $t=t'r$ where $t' \in M_G$ and $r \in \bar R$ is \emph{succinctly represented}
by a pair $(P,r)$ where $P$ is an SLP such that $\sem{P} = t'$.
We remark that the size of the term $r$ is polynomially bound by the size of the program
and therefore can be represented explicitly.

For small terms, i.e., terms in $S$, on the other hand, we cannot hope for unique factorizations.
Since the size
of each small term is bound by the size of the program, each small term $s \in S$ is \emph{succinctly represented}
by a pair $(P,s)$ where $P$ is an SLP such that $\sem{P} = \bullet$.

Similar as for initialization-restricted programs, during the weakest pre-condition calculation,
we assume that each occurring term is succinctly represented.
Let us again consider the operation substitution.
In order to obtain polynomial algorithms,
we must ensure that substitution of succinctly represented terms is polynomial.
Consider the non-ground terms $s\x,t\y$ where $s,t \in M_G$.
Then the succinct representation of the resulting term $(s\x)[t\y/\x]$ is determined
in a similar way as for initialization-restricted programs, and therefore can be constructed in
polynomial time.
Now consider the terms $s\x,t$ where $s \in M_G$ and $t \in T$ is ground.
Then the resulting term of the substitution $(s\x)[t/\x]$ is given as $st$.
If the term is large, then in order to succinctly represent $st$,
the unique factorization must be determined in polynomial time.
\begin{lemma}\label{l:succgroundconcat}
  Given succinctly represented terms $s,t$ where $s \in M_G$ and $t \in T$.
  Then a succinct representation of $st \in T$ can be determined in time
  polynomial in the size of a maximal element in $\bar R$.
\end{lemma}
\begin{proof}
  First assume that $t \in L$ is \emph{large}. This means that $t$ is represented by a pair $(Q,r)$
  where $Q$ is an SLP for some term $t'\in M_G$ and $r$ is a term in $\bar R$.
  Then the unique factorization of $st$ is given by $s'r$ where $s' = st'$ --- for which an SLP can
  be constructed from an SLP for $s$ and $Q$ by introducing one fresh unknown together with a single definition.

  Finally, assume that the term $t$ is \emph{small}.
  Given an SLP $P$ for the term $s$, our goal is to determine the unique factorization $st =s'r$ with $s'\in M_G$ and $r\in\bar R$.
  If $s=\bullet$, nothing must be done.
  Otherwise, assume that $s$ is given as the factorization $s_1 \cdots s_k$.
  Then we consider the factorization $s_1 \cdots s_{k-1} s'_k$
  where $s'_k = s_k[t/\bullet]$. This factorization equals the term $st$.
  If $k=1$, we are done.
  If $k>1$ and the term $s'_k = s_k[t/\bullet]$ is contained in the set $\bar R$ of minimally large terms,
  i.e., $s'_k$ is not a small term, then we have found the unique factorization of $st$.
  Otherwise, we proceed by constructing $s'_{k-1}=s_{k-1}[s'_k/\bullet]$
  and so on, until either we exhausted the factors of $s$ or obtained the factorization
  $st=s' s'_{k-h}$ where $s'= s_1\cdots s_{k-h-1}$ and $s'_{k-h}= s_{k-h}\cdots s_k t\in\bar R$.
  Since the size of every term in $\bar R$ is bounded by the size of the input program, so is the number $h$.
  For every length $h\leq h'\leq k$, SLPs for the intermediately occurring prefixes of $s$ can
  be determined in time $\mathcal{O}(d)$ by Theorem~\ref{t:termops},
  if $d$ is the depth of the SLP for $s$.
\end{proof}

The previous Lemma~\ref{l:succgroundconcat} enables us to state the following two lemmas:
\begin{lemma}\label{l:succ_subst}
  Consider a single equality $As_1 \doteq Bs_2$ or $As \doteq C$ where $s_1,s_2,s \in T \cup T'$ are succinctly represented.
  Then a succinct representation of the weakest pre-condition of the equality w.r.t.\ an assignment
  $\x \Let t$ can be determined in time polynomial in the size of $t$ and in the size of a maximal element in $\bar R$.
  \qed
\end{lemma}
\begin{lemma}\label{l:succ_adapt}
  Consider a single equality $As\x \doteq Bt\y$ (resp.\ $As \doteq Bt\x$, $As\x \doteq Bt$, or $As\x \doteq C$)
  where the occurring terms $s,t \in M_G$ are succinctly represented.
  Moreover, assume that each term of type $T \cup T'$ occurring in the weakest pre-condition $\phi$ of a generic post-condition
  $A\x \doteq B\y$ (resp.\ $A\x \doteq C$) w.r.t.\ a procedure call $p()$ is also succinctly represented.
  Then a succinct representation of the weakest pre-condition of the equality w.r.t.\ a procedure call $p()$
  can be computed in time polynomial in the number of equalities in $\phi$ and in the size of a maximal element in $\bar R$.
  \qed
\end{lemma}
The proofs of the lemmas are analogous to the proofs of Lemma~\ref{l:ir-succ_subst} and~\ref{l:ir-succ_adapt} except
that for the substitution we also need Lemma~\ref{l:succgroundconcat}.

In order to compute solutions in polynomial time for the constraint systems \textbf{S} and \textbf{R},
$T$-subsumption for one-variable and approximate $T$-subsumption for two-variable
equalities must be decidable in polynomial time.
\begin{theorem}\label{t:succonevarsubsumption}
  For finite sets $E,E'$ of equalities of the form $As \doteq C$ where $s \in T \cup T'$ are succinctly represented
  it is decidable in polynomial time whether $\bigwedge E$ $T$-subsumes $\bigwedge E'$ or not.
\end{theorem}
\begin{proof}
  Consider two distinct equalities $As\x \doteq C$ and $At\x \doteq C$.
  If the conjunction of them is satisfiable, then $s = wu$ and $t = wv$
  for some $u,v,w \in M_G$
  such that $w$ is a longest common prefix of $s,t$ and $u \neq v$ but $u\x = v\x$ must hold.
  According to Theorem~\ref{t:termops} the longest common prefix of two succinctly
  represented terms can be determined in polynomial time. Similar representations for $u,v$
  can be determined in polynomial time, too.
  Assume that the sizes of the terms $u,v$ are not bound by the maximal size of an element in $\bar R$,
  then the terms $u\x,v\x$ are large terms no matter what ground term the variable $\x$ is actually bound to.
  But then the terms $u,v$ must have a common prefix which is
  a contradiction to the assumption that $w$ is the longest common prefix of $s,t$ if the conjunction is satisfiable.
  Therefore, assume that the sizes of the terms $u,v$ are bound by the maximal size of an element in $\bar R$.
  Then the most general unifier of $u\x = v\x$ can be determined in polynomial time.
  Assume the most general unifier maps $\x$ to the ground term $t' \in S$. Then the initial
  conjunction is equivalent to the conjunction of the equalities $As\x \doteq C$ and $Att' \doteq C$
  where the latter equality does not contain any program variable.
  According to Lemma~\ref{l:succgroundconcat} a succinct representation of the term $tt'$ can be determined
  in polynomial time. Overall, the equivalent conjunction can be determined in polynomial time.

  For equalities which contain no program variable we have the following result.
  Consider two equalities $As \doteq C$ and $At \doteq C$ where $s,t \in T$ are ground.
  If $s = t$, then one equality subsumes the other. Otherwise, if $s \neq t$, then
  the conjunction of them is unsatisfiable. For succinctly represented terms such equality checks
  can be performed in polynomial time from which the assertion of the theorem follows.
\end{proof}

\begin{theorem}\label{t:succgroundsubsumption}
  For finite sets $E,E'$ of equalities of the form $As \doteq Bt$ where $s,t \in T \cup T'$ are succinctly represented,
  it is decidable in polynomial time whether $\bigwedge E$ approximately $T$-subsumes $\bigwedge E'$ or not,
  whenever $A$ or $B$ equals $\bullet$.
\end{theorem}
\begin{proof}
  \textbf{Ground equalities:} Let us first consider only equalities of the form $As \doteq Bt$ where
  $s,t \in T$ are ground.
  Then in the following we assume that each conjunction of equalities is not trivially unsatisfiable, i.e.,
  there exist no two equalities of the form $As \doteq Bt$ and $As \doteq Bt'$ where $t \neq t'$,
  or vice versa, where the roles of $A$ and $B$ are interchanged.
  If two succinctly represented terms in $T$ are equal or not, is decidable in polynomial time.

  First consider equalities of the form $As \doteq Bt$ where $s,t \in L$ are large terms.
  The proof is analogous to the corresponding proofs for Theorem~\ref{t:ir-succgroundsolve}
  where the set $T$ is replaced with the set $L = M_G \bar R$, i.e.,
  instead of the set $R$ we rely on the set $\bar R$ of unique end marker terms.

  Now consider three equalities $As_i \doteq Bt_i$ where $s_i \in L$ are large terms and
  $t_i \in S$ are small terms for $i=1,2,3$.
  For the proof of this case, we require to extend the notion of substitution to
  a replacement of occurrences of arbitrary subterms.
  Consider arbitrary ranked terms $s,t,t' \in \T_\Omega(X \cup \{\bullet\})$. Then by $s[t/t']$
  we denote the term where all occurrences of $t'$ in $s$ are replaced by the term $t$.
  Formally, if $s$ does not contain the subterm $t'$, then $s[t/t'] = s$.
  Otherwise, if $s$ contains the subterm $t'$, then let $s = s't'$ such that $s' \in \C_\Omega$ does not contain the subterm $t'$.
  Then $s[t/t'] = s't$.

  We then proceed as follows. Assume that there exist $i,j \in [1,3]$ such that $t_i$ does not occur in $s_j$.
  If $i=j$, then the single equality $As_i \doteq Bt_i$ is not satisfiable.
  Therefore assume now that $i \neq j$. If the conjunction of the three equalities is satisfiable,
  then the solution for $B$ must not contain occurrences of $t_i$,
  i.e., $u = s_i[\bullet/t_i]$ is the only possible solution for $B$.
  If all three equalities are satisfied
  by this solution, then the first two would already have $B=u$ as their unique solution.
  Accordingly, the third equality is subsumed. Whether or not $B=u$ is a solution can be decided in polynomial time.

  In the following we therefore assume that for each $i,j \in [1,3]$ the term $t_i$
  occurs at least once in the term $s_j$.
  We define an equivalence relation of terms as follows.
  Let $\hash$ denote a fresh symbol and let $s,s',t,t' \in T$.
  If $t,t'$ are incomparable, i.e., there exists no $u \in M_G$ such that $t = u t'$ or
  $t' = u t$, then the terms $s,s'$ are equivalent modulo the terms $t,t'$ if
  $s[\hash/t,\hash/t'] = s'[\hash/t,\hash/t']$ holds. Otherwise, if there exists a $u \in M_G$
  such that $t = u t'$, then the terms $s,s'$ are equivalent modulo the terms $t,t'$ if
  $(s[\hash/t])[\hash/t'] = (s'[\hash/t])[\hash/t']$ holds. The case where $t' = u t$ holds is similar.
  For all three cases we can decide which term to substitute first by comparing the size of both terms $t,t'$.
  That means, if $t = u t'$ (resp.\ $t' = u t$) holds, then $\sizeT{t} > \sizeT{t'}$ (resp. $\sizeT{t'} > \sizeT{t}$)
  must hold, too. In case the terms are incomparable it does not matter in which order we substitute the
  terms. Assume $\sizeT{t} \geq \sizeT{t'}$, then the terms $s,s'$ are equivalent modulo the terms
  $t,t'$ if $(s[\hash/t])[\hash/t'] = (s'[\hash/t])[\hash/t']$ holds, which we denote by $(s = s') \bmod {t,t'}$.
  We extend the equivalence relation as follows. Let $t'' \in T$ and assume that $\sizeT{t} \geq \sizeT{t'} \geq \sizeT{t''}$, then
  $s$ and $s'$ are equivalent modulo the terms $t,t',t''$ if
  $((s[\hash/t])[\hash/t'])[\hash/t''] = ((s'[\hash/t])[\hash/t'])[\hash/t'']$ holds which we denote by
  $(s = s') \bmod {t,t',t''}$.
  We observe that if the conjunction $As_1 \doteq Bt_1 \land As_2 \doteq Bt_2$ is
  satisfiable, then $s_1,s_2$ differ only in some occurrences of $t_1,t_2$.
  That means that $(s_1 = s_2) \bmod {t_1,t_2}$ must hold.
  A similar argument holds for the conjunction $As_1 \doteq Bt_1 \land As_3 \doteq Bt_3$
  and for the conjunction $As_2 \doteq Bt_2 \land As_3 \doteq Bt_3$.
  Observe that the other direction does not necessarily hold, i.e., if the conjunction is not satisfiable,
  then $(s_1 \neq s_2) \bmod {t_1,t_2}$ need not hold. For example, consider the conjunction
  $Af(a,b) \doteq Ba \land Af(b,a) \doteq Bb$ which is not satisfiable but
  $f(a,b)[\hash/a,\hash/b] = f(\hash,\hash) = f(b,a)[\hash/a,\hash/b]$ holds.
  However, we claim that if
  \begin{align}
    &(s_1 = s_2) \bmod {t_1,t_2} \label{eq:s12}\\
    &(s_1 = s_3) \bmod {t_1,t_3} \label{eq:s13}\\
    &(s_2 = s_3) \bmod {t_2,t_3} \label{eq:s23}\\
    &(s_1 = s_2 = s_3) \bmod{t_1,t_2,t_3} \label{eq:s123}
  \end{align}
  holds, then the conjunction $As_1 \doteq Bt_1\land As_2 \doteq Bt_2\land As_3 \doteq Bt_3$
  is satisfiable. Since for $A=\bullet$, the two first equalities uniquely determine the solution for $B$,
  we conclude that the third equality is subsumed. Our claim is proved as follows.

  In the following we denote by $s|_p = t$ that a term $s \in T$ contains at position $p$ the subterm $t \in T$.
  From~\eqref{eq:s123} follows that all three terms $s_1,s_2,s_3$ share a common pattern $u'$ which is obtained by successively
  replacing all subterms $t_1,t_2,t_3$ with $\hash$, where we proceed from the larger to the smaller terms.
  From $u'$, we then construct a solution $u$ for $B$ by replacing the occurrences of $\hash$ in $u'$ with
  terms $t_1,t_2,t_3$ or $\bullet$.
  Let $p$ be any position of a leaf $\hash$ in $u'$.
  \begin{align}
    &\text{If } s_1|_p = s_2|_p    \text{ then let } u|_p = s_1|_p  \label{eq:sol1}\\
    &\text{If } s_1|_p \neq s_2|_p \text{ then let } u|_p = \bullet \label{eq:sol2}
  \end{align}
  We claim that the resulting term $u$ is indeed a solution for $B$, which satisfies all three equalities.
  If $u|_p = t_1$ then according to~\eqref{eq:sol1} $s_1|_p = s_2|_p = t_1$ and from~\eqref{eq:s23} follows that $s_3|_p = t_1$.
  If $u|_p = t_2$ then according to~\eqref{eq:sol1} $s_1|_p = s_2|_p = t_2$ and from~\eqref{eq:s13} follows that $s_3|_p = t_2$.
  Otherwise, assume that $u|_p = t_3$. Then according to~\eqref{eq:sol1} $s_1|_p = s_2|_p = t_3$.
  If $s_3|_p = t_1$, then~\eqref{eq:s23} implies that $s_2|_p = t_1$ which is a contradiction.
  Similarly, if $s_3|_p = t_2$, then~\eqref{eq:s13} implies that $s_1|_p = t_2$ which again is a contradiction.
  Therefore, $s_3|_p = t_3$ must hold.
  In total we have, if $u|_p = t_i$, then $s_1|_p = s_2|_p = s_3|_p = t_i$ for $i=1,2,3$.
  Now consider the case where $u|_p = \bullet$.
  Assume that $s_1|_p = t_2$, then from~\eqref{eq:s12} and~\eqref{eq:sol2} follows $s_2|_p = t_1$.
  However, then from~\eqref{eq:s13}
  follows that $s_3|_p = t_2$ and from~\eqref{eq:s23} follows that $s_3|_p = t_1$ which is a contradiction.
  Hence, $s_1|_p = t_1$ and $s_2|_p = t_2$ must hold.
  A similar argument holds for $s_3|_p = t_3$ from which
  we conclude that $s_i = u t_i$ for $i=1,2,3$. Therefore, $B=u$ is indeed a solution satisfying all three equalities.
  This complete the proof of our claim.

  What remains to prove is that the equality checks $(s_i = s_j) \bmod {t_i,t_j}$
  and $(s_i = s_j) \bmod {t_1,t_2,t_3}$ can be done in polynomial time.
  For that we must show that from an arbitrary succinctly represented large term,
  a succinctly represented and uniquely factorized term can be derived where certain small terms are substituted by a fresh
  symbol. We explain the idea for the test $(s_1 = s_2) \bmod {t_1,t_2}$.
  W.l.o.g.\ let $\sizeT{t_1} \geq \sizeT{t_2}$ and $\sigma = [\hash/t_1][\hash/t_2]$.
  Let $G' = \Set{g\sigma | g \in G}$ and $R' = \Set{r\sigma | r \in R}$.
  We extend the factorization of terms in $T = M_G R$ to terms in $T' = M_{G'} R'$.
  In Section~\ref{s:general} we have partitioned the set of terms $T$ into non-uniquely factorizable small terms $S$ and
  uniquely factorizable large terms $L$, i.e., $T = M_G R = S \cupplus L$. We proceed along the same line for $T'$ which we partition
  into $\hash$-small terms $S'$ which are non-uniquely factorizable, and into $\hash$-large terms $L'$
  which are uniquely factorizable. The set $S'$ equals then the set ${(G' \cup R')}^*$ where $*$ is the subterm closure,
  and the set $L'$ equals the set $M_{G'} R' \setminus S'$.
  We call a term minimally $\hash$-large if it is a minimal term in $L'$.
  The (finite) set of all minimally $\hash$-large terms is denoted by $\bar R'$.
  Then every $\hash$-large term $s'\in L'$ can be uniquely factored
  into $s'=u'r'$ where $u'\in M_{G'}$ and a term $r' \in \bar R'$ which is minimally $\hash$-large.
  If one of the terms $s_1\sigma$ or $s_2\sigma$ is not $\hash$-large, then the size of that term is polynomial. Therefore,
  the equality test can be realized in polynomial time as well.
  Accordingly assume that the terms $s_1\sigma$ and $s_2\sigma$ are both $\hash$-large.
  In this case, our goal is to determine from the succinct representations of the factorizations of $s_1,s_2$,
  succinct representations for the factorizations of $s_1\sigma,s_2\sigma$ which then can be compared in polynomial time.
  For that, consider a factorization $s = u_1 \cdots u_k r$ of a large term $s\in T$ into irreducible factors $u_i\in M_G$
  and a minimally large term $r \in \bar R$, and assume that $s'=s\sigma$ is $\hash$-large.
  Then there is a maximal index $j$ such that $r'_j = (u_j \cdots u_k r)\sigma$ is $\hash$-large.
  This index can be found in polynomial time. Moreover, $r'_j$ can then be uniquely factored in
  polynomial time into $r'_j = u'r'$ for a minimally $\hash$-large term $r' \in \bar R'$ and $u'\in M_{G'}$.
  Then the unique factorization of $s'$ is given by:
  \[
  s'= v'_1 \cdots v'_{j-1} u' r'
  \]
  where for each $i$, $v'_i$ is a factorization of $u_i\sigma$ into irreducible factors in $M_{G'}$.
  Note that the \emph{lengths} of the factorizations $v'_i$ are bounded by the sizes of the corresponding factors and
  thus of the sizes of right-hand sides of the input program. Therefore these factorizations can be obtained in polynomial
  time as well. These factorizations then allow us to construct from an SLP for $u_1 \cdots u_{j-1}$, an SLP for
  $v'_1 \cdots v'_{j-1} u'$.
  Altogether, we obtain a succinct representation for $s'$ from a succinct representation of $s$ in
  polynomial time from which the assertion of this part follows.

  A similar argument holds for equalities of the form $As \doteq Bt$ where $s \in S$ is small and $t \in L$ is large.

  \textbf{Non-ground Equalities:} Now we consider equalities which contain at least one program variable.
  We first prove that $L$-subsumption for finite conjunctions of equalities of
  the same format is decidable in polynomial time.

  For equalities of the formats $[F_{\x,\y}], [F_{\x,\cdot}], [F_{\cdot,\x}]$ the proofs are
  analogous to the corresponding proofs of Theorem~\ref{t:ir-tsubptime} where the set $T$ is replaced
  with the set $L = M_G \bar R$, i.e., instead of the set $R$ we rely on the set $\bar R$ of unique
  end marker terms.

  Now consider two equalities $As \doteq Bt\x$ and $As \doteq Bt'\x$ where $s \in S$ is a small term, i.e.,
  equalities of the format $[F_{s,\x}]$. Then the first equality $L$-subsumes the second equality, if $t = t'$ holds.
  This is decidable in polynomial time. Otherwise, the conjunction is $L$-unsatisfiable.
  A similar argument holds for two equalities of the format $[F_{\x,s}]$.

  We conclude that $L$-subsumption for equalities of the same format is decidable in polynomial time.
  In order to decide $T$-subsumption between conjunctions of sets $E,E'$ of equalities of the same format,
  for each small substitution $\sigma$, $L$-subsumption between $E\sigma$ and $E'\sigma$ has to be decided.
  Since there exist at most polynomial many small substitutions and formats of equalities,
  we conclude that approximate $T$-subsumption is decidable in polynomial time.
\end{proof}
We showed that solutions to the constraint systems {\bf S} and {\bf R} can be determined in polynomial time.
For \emph{initialization-restricted} programs we also showed that a succinct representation of all solutions
can be determined in polynomial time. Whereas for \emph{unrestricted} programs we show that given a
candidate solution for the template variables $A$ and $B$ where at least one equals $\bullet$,
it is decidable in polynomial time whether or not the solution holds.
\begin{theorem}\label{t:succgroundverify}
  Given a term $u \in \C_\Omega$ and an equality $As \doteq Bt$ where $s,t \in T$ are ground and succinctly represented.
  Then it is decidable in time polynomial in the size of the term $u$ and in the size of a maximal term in $S$,
  whether or not $A=u$ and $B=\bullet$, or vice versa, $A=\bullet$ and $B=u$ is a solution for the equality.
\end{theorem}
\begin{proof}
  Let us first consider the case for $A=u$ and $B=\bullet$, i.e., decide if $us = t$ holds or not.

  Assume that $s,t \in L$ are large terms.
  If $u \in M_G$, i.e., all subterms of $u$ are small, then $u$ must be a prefix of $t$, and $s$ must be a suffix of $t$, i.e.,
  $us = t$ must hold. This is decidable in time polynomial in the size of $u$ and polynomial
  in the sizes and lengths of the SLPs representing $s,t$.
  Otherwise, if $u$ contains large terms as subterms, i.e., $u \in \C_\Omega \setminus M_G$.
  Then $u=vw$ for some $v \in M_G$ and some irreducible element $w \in \C_\Omega \setminus M_G$ must hold.
  Furthermore, $t=vw's$ for some irreducible element $w' \in M_G$ such that $w$ equals $w'$ where some occurrences of $\bullet$
  are substituted by $s$ must hold. This is decidable in time polynomial in the size of $u$
  and polynomial in the lengths of the SLPs representing $s,t$.

  Now consider the case where $s \in S$ is small and $t \in L$ is large. Then $us = t$ is
  decidable in time polynomial in the size of $u$ and in the size of $s$ which is bound by
  the size of the program.

  Otherwise, if $s \in L$ is large and $t \in S$ is small, then the equality is not satisfiable.

  Now assume that both $s,t \in S$ are small. Whether or not $us$ is a small term and if $us$ equals $t$
  is decidable in polynomial time.

  Furthermore, verifying if $A=\bullet$ and $B=u$ is a solution for the equality is similar
  from which the assertion of this theorem follows.
\end{proof}
Finally this enables us to state our main result for unrestricted programs and one- or two-variable equalities:
\begin{theorem}
  Assume that $p$ is a program where all right-hand sides of assignments contain at most one variable.
  Then for every program point $u$ of $p$ and every equality of the form
  $\x \doteq t$ where $t\in T\cup T'$, it can be verified in time polynomial
  in the size of the program as well as the size of $t$ whether or not the equality is an invariant.
  \qed
\end{theorem}

Recall that
for \emph{initialization-restricted} programs, each possible run-time value can be uniquely factorized.
This property enabled us to derive in polynomial time from a ground equality $As \doteq Bt$ where $s,t \in T$
all possible solutions for the template variables $A$ and $B$ where at least one equals $\bullet$.
Consider the case where $A=\bullet$ and assume that $s=uvt$ for some $u,v \in M_G$ where $v$ is an irreducible element.
Then each solution for $B$ has $u$ as a prefix --- which might be exponentially large. That means, that the
solutions only differ in the very last factor which can be derived from the element $v$.
Accordingly, we were able to provide a succinct characterization of \emph{all} solutions.
The situation is more complicated for \emph{unrestricted} programs.
For these, only weaker forms of factorization are available. Thus,
substitutions of right-hand sides may still result in terms which are still \emph{small} and
therefore cannot be uniquely factorized.
\begin{example}
  Assume that $a,b \in S$ are small terms and $r \in \bar R$ is a minimally large term.
  Then consider the uniquely factorized equalities
  \begin{align*}
    A \enspace f(\bullet,\bullet) \enspace g(a,h(b,\bullet,a),h(b,\bullet,b)) \enspace r &\doteq B \enspace a \\
    A \enspace f(g(a,\bullet,\bullet),g(b,\bullet,\bullet)) \enspace h(b,\bullet,b) \enspace r &\doteq B \enspace b
  \end{align*}
  Since $a$ and $b$ are small terms and the template variable $A$ is applied to large terms,
  $B$ cannot equal $\bullet$ in any possible solution.
  Therefore, now assume that $A=\bullet$.
  Then the unique solution for $B$, satisfying both equalities, equals
  \[
    f(g(a,h(b,r,\bullet),h(b,r,b)),g(\bullet,h(b,r,\bullet),h(b,r,b)))
  \]
  Thus, all three factors from the original equality are collapsed into a single irreducible term for $B$.
  This irreducible term contains the large term $h(b,r,b)$ as a subterm and is contained in $\C_\Omega \setminus M_G$.
  \qed
\end{example}
From the previous example we conclude that,
in contrast to \emph{initialization-restricted} programs,
we have for \emph{unrestricted} programs that a solution
is not necessarily in $M_G$ but might very well also be in $\C_\Omega \setminus M_G$.
The solution need not reflect the factorizations of the terms of the initial equalities.
The factorization of terms, however, was the basis of our compression scheme via SLPs.
Accordingly, it remains unclear how to derive compressed representations of solutions in polynomial time.

 \subsection{Complexity results for Verifying Multi-Variable Equalities}

Let us now consider a multi-variable invariant candidate such as $\x \doteq f(g\y,\z)$.
In this case, the right-hand side $f(g\y,\z)= t[\y,\z]$ where $t$ is the (multi-variable) pattern
$t= f(g\bullet_1,\bullet_2)$ for distinct variables $\bullet_1,\bullet_2$.
Now consider a generic post-condition $\x' \doteq f(gA\y',B\z')$ which might occur during
the proof that the given equality indeed is an invariant at some program point.
In contrast to the pattern, the terms which may be substituted into one of the  program variables or
the template variables $A,B$ of the right-hand side during the fixpoint iteration may grow exponentially
deep and therefore should be succinctly represented.
Now consider a term which is substituted into the left-hand side.
For this term, the root must be deconstructed according to the constructors occurring in $t$.
This deconstruction can also be realized for succinctly represented terms in polynomial time.

For the multi-variable case we observe that during the {\bf WP} computation we obtain for a post-condition a
conjunction possibly containing one-, two-, and multi-variable equalities.
A conjunction of two multi-variable equalities which coincide in the left-hand side and the pattern of the right-hand side
is equivalent to a conjunction of one of them and polynomial many one- and two-variable equalities.
Such an equivalent conjunction can be determined in time polynomial in the size of the invariant candidate.
Since for conjunctions of one- and two-variable equalities approximate $T$-subsumption is decidable in polynomial time,
approximate $T$-subsumption is also decidable in polynomial time for conjunctions containing multi-variable equalities,
i.e., we have proven the following lemma:
\begin{lemma}
  For finite sets $E,E'$ of one-, two-, and multi-variable equalities where each term in $T \cup T'$
  is succinctly represented, it is decidable in polynomial time whether $\bigwedge E$ approximately
  $T$-subsumes $\bigwedge E'$ or not.
  \qed
\end{lemma}
We note that from a single invariant candidate $\x \doteq t$ where $t \in \T_\Omega(\X)$,
exponentially many generic multi-variable post-conditions can be derived, i.e.,
we have exponentially many different formats of multi-variable equalities.
Still, we have:
\begin{theorem}
  Assume that $p$ is a program where all right-hand sides of assignments contain at most one variable.
  Then for every program point $u$ of $p$ and every multi-variable Herbrand equality
  $\x \doteq t$ where $t \in \T_\Omega(\X)$ has at most $k$ variables, it can be verified in time polynomial
  in the size of the program as well as the size of $t$, and exponential only in $k$
  whether or not the equality is an invariant.
  \qed
\end{theorem}
 \section{Conclusion}\label{s:conclusion}

We have provided an analysis which infers all inter-procedurally valid Herbrand equalities
for programs where all assignments are taken into account whose right-hand sides depend on at most one variable.
The novel analysis is based on three main ideas.
First, we restricted general satisfiability, subsumption
and equivalence to satisfiability, subsumption and equivalence w.r.t.\ a
set of values subsuming all possible run-time values of a given program.
Together with our factorization theorem, this allowed us to apply the monoidal methods from~\cite{Gulwani07}
to effectively infer all inter-procedurally valid two-variable Herbrand equalities,
at least for programs, which we called \emph{initialization-restricted}.
In the second step, we abandoned this restriction by introducing the
extra distinction between \emph{large} values (which can be uniquely factored) and \emph{small} ones (of which
there are only finitely many).
Finally, we showed how general Herbrand equalities could be handled.
In Section~\ref{s:polytime} we then provided a polynomial-time algorithm which infers all
two-variable Herbrand equalities for \emph{initialization-restricted} programs.
For \emph{unrestricted} programs, we were at least able to verify in polynomial time
whether or not a given equality is an invariant at a given program point.
This algorithm could also be extended to \emph{general} Herbrand equalities (possibly containing more
than one two variables).

Still, it remains open whether general Herbrand invariants can be inferred also for programs where
right-hand sides may contain more than one variable.

\subsubsection*{Acknowledgments.}
The authors would like to thank the anonymous reviewers for their valuable comments and suggestions.
 \appendix

\bibliographystyle{abbrv}

\section{Global and Local Program Variables}\label{a:locals}

In this appendix we indicate how our method can be extended in order to also deal with programs which contain global
as well as local variables. For this we first extend our program model from Section~\ref{s:programs} as follows.
We now assume that the (finite) set of program variables $\X$ contains a subset $\mathbf{L} \subseteq \X$
consisting of \emph{local} program variables, while the remaining variables are considered as \emph{global}.
The scope of local variables is meant to be restricted to the body of the current procedure.
At the start of a procedure call, the fresh local variables are assumed to be \emph{uninitialized}, i.e., have
any value, whereas at procedure exit, the current locals are abandoned while the locals of the calling procedure
are recovered.
By means of global variables, this simple model already allows to realize call-by-value variable passing as well as
the returning of functional results.

In order to deal with a non-empty set of locals, we enhance the weakest pre-condition calculus by an operator
$\mathcal{H}$ which takes the {\bf WP}-transformation realized by the body of a procedure as an argument, and
returns the {\bf WP}-transformation of the procedure call.
For a given {\bf WP}-transformation $f$, the {\bf WP}-transformation $\mathcal{H}(f)$ is defined as follows.
\[
\begin{array}{lll@{\quad}l}
\mathcal{H}(f)(A\x \doteq C) &=& \phantom{(}\forall\textsf{L}.\,f(A\x \doteq C)
			&\text{$\x$ global}		\\
\mathcal{H}(f)(A\x \doteq B\y) &=& \phantom{(}\forall\textsf{L}.\,f(A\x \doteq B\y)
			&\text{$\x,\y$ global}		\\
\mathcal{H}(f)(A\x \doteq B\y) &=& (\forall\textsf{L}.\,f(A\x \doteq C))[B\y/C]
			&\text{$\x$ global, $\y$ local}		\\
\mathcal{H}(f)(A\x \doteq B\y) &=& (\forall\textsf{L}.\,f(A\y \doteq C))[B/A,A\x/C]
			&\text{$\x$ local, $\y$ global}		\\
\mathcal{H}(f)(e)	&=& \phantom{(}e	&\text{$e$ contains no globals}	\\
\end{array}
\]
where $\textsf{L}$ is the sequence of local variables in $\mathbf{L}$.
Accordingly, the constraints for call edges in the constraint systems {\bf S} and {\bf R}
must be changed into:
\begin{align*}
\semT{u} &\implies \mathcal{H}(\semT{s_p}) \circ \semT{v} && \text{for each $(u,p(),v) \in E$}
\intertext{and}
\semR{v} &\implies \semR{u} \circ \mathcal{H}(\semT{s_p}) && \text{for each $(u,p(),v) \in E$}
\end{align*}
respectively.
 
\end{document}